\documentclass[conference,letterpaper]{ieeeconf} 

\IEEEoverridecommandlockouts                              

\usepackage{graphicx,amsmath,amsbsy,amssymb,url}
\usepackage{colortbl}	
\usepackage{tabularx}
\usepackage{theorem}	
\usepackage{algorithm,algorithmic}
\newcommand{\R}{{\mathbb R}}

\newcommand{\N}{{\mathcal{N}}}

\theorembodyfont{\upshape}
\newtheorem{theorem}{Theorem}
\newtheorem{assumption}{Assumption}
\newtheorem{lemma}{Lemma}

\newtheorem{definition}{Definition}
\newtheorem{corollary}{Corollary}

\title{An Integrated Design of Optimization and Physical Dynamics for Energy Efficient Buildings: A Passivity Approach}

\author{Takeshi Hatanaka, Xuan Zhang, Wenbo Shi, Minghui Zhu and Na Li
\thanks{T. Hatanaka is with School of Engineering, Tokyo Institute of Technology, Tokyo 152-8552, JAPAN
(hatanaka@ctrl.titech.ac.jp).
X. Zhang, W. Shi and N. Li are
with Electrical Engineering and Applied Mathematics of the School of Engineering and Applied Sciences,
Harvard Univ., 33 Oxford St, Cambridge, MA 02138, USA.
M. Zhu is with Department of Electrical Engineering,
Pennsylvania State Univ.,
University Park, PA 16802, USA.
}}

\begin{document}
\maketitle
\thispagestyle{empty}
\pagestyle{empty}

\begin{abstract}
In this paper, we address energy management for
heating, ventilation, and air-conditioning (HVAC) systems in buildings, 
and present a novel combined optimization and control approach.
We first formulate a thermal dynamics and an associated optimization problem. 
An optimization dynamics is then designed based on a standard primal-dual algorithm, and its 
strict passivity is proved.
We then design a local controller and prove that the physical dynamics
with the controller is ensured to be passivity-short.
Based on these passivity results, we interconnect the optimization and
physical dynamics, and prove convergence of the room temperatures to
the optimal ones defined for unmeasurable disturbances.
Finally, we demonstrate the present algorithms through simulation.
\end{abstract}

\section{Introduction}

Stimulated by strong needs for reducing energy consumption of buildings,
smart building energy management algorithms have been developed both in industry and academia.
In particular, about half of the current consumption is known to be occupied by
heating, ventilation, and air-conditioning (HVAC) systems, and a great deal of works
have been devoted to HVAC optimization and control \cite{survey}.
In this paper, we address the issue based on a novel approach
combining optimization and physical dynamics.

Interplays between optimization and physical dynamics 
have been most actively studied in
the field of Model Predictive Control (MPC), 
which has also been applied to building HVAC control 
\cite{survey}--\cite{MPC4}.
While the MPC approach regards the optimization process as a static map from
physical states to optimal inputs, 
another approach to integrating optimization and physical dynamics
is presented in \cite{anu}--\cite{SPV_16} 
mainly motivated by power grid control. 
There, the solution process of the optimization is viewed as a dynamical system,
and the combination of optimization and physical dynamics is
regarded as an interconnection of dynamical systems.
The benefits of the approach relative to MPC are as follows. 
First, the approach allows one to avoid complicated modeling and prediction 
of factors hard to know in advance,
while MPC needs their models to predict future system evolutions.
Second, since the entire system is a dynamical system, its stability and performance
are analyzed based on unifying dynamical system theory. 

%

In this line of works, 
Shiltz et al. \cite{anu} addresses smart grid control,
and interconnects a dynamic optimization process
and a locally controlled grid dynamics. 
The entire process is then demonstrated through simulation.
The authors of \cite{TAC14, ATN14} incorporate the grid
dynamics into the optimization process by
identifying the physical dynamics with a subprocess of seeking the optimal solution.
A scheme to eliminate structural constraints required in \cite{TAC14, ATN14}
is presented by Zhang et al. \cite{CDC15} while instead assuming measurements of disturbances.
A similar approach is also taken for power grid control
in Stegink et al. \cite{SPV_16}.

In this paper, we address integrated design of optimization and physical
dynamics for HVAC control based on passivity, where we regard the optimization process as 
a dynamical system similarly to \cite{anu}--\cite{SPV_16}.
Interconnections of such dynamic HVAC optimization with a building dynamics are partially studied in 
\cite{Xuan}, where temperature data for all zones and their derivatives in the physics side
are fed back to the dynamic optimization process to recover the disturbance terms.
However, such data are not always available in practical systems.
We thus present an architecture relying only on temperature data
of a subgroup of zones with HVAC systems.

The contents of this paper are as follows.
A thermal dynamic model with unmeasurable disturbances and an associated optimization problem are first presented.
We then 
formulate an optimization dynamics based on the primal-dual gradient algorithm \cite{cortes}.
The designed dynamics is then proved to be strictly passive from a transformed disturbance estimate
to an estimated optimal room temperature.
We next design a controller so that the actual room temperature tracks a given reference, 
and produces a disturbance estimate.
Then, the physical dynamics is proved to be passivity-short from the reference to the disturbance estimate.
From these two passivity-related results, we then interconnect the 
optimization and physical dynamics, and prove convergence of the actual room temperature
to the optimal solution defined for the unmeasurable actual disturbance.
Finally, the presented algorithm is demonstrated through simulation.

\section{Problem Settings}
\label{sec:2}

\subsection{Preliminary}

In this section, we introduce the concept of passivity.
Consider a system with a state-space representation
\begin{eqnarray}
\dot x = \phi(x,u),\ \ y = \varphi(x,u),
\label{eq:0.0}
\end{eqnarray}
where $x(t)\in \R^n$ is the state, 
$u(t) \in \R^p$ is the input and $y(t)\in \R^p$ is the output.
Then, passivity is defined as below.
\begin{definition}
The system (\ref{eq:0.0}) is said to be passive if there exists a 
positive semi-definite
function $S: \R^n \to \R_+ := [0, \infty)$, 
called storage function, such that 
\begin{eqnarray}
S(x(t)) - S(x(0)) \leq \int^{t}_0 y^T(\tau)u(\tau)d\tau
\label{eq:0.1}
\end{eqnarray}
holds for all inputs $u:[0, t]\to \R^p$, 
all initial states $x(0)\in \R^n$ and all $t \in \R^+$.
In the case of the static system $y = \varphi(u)$,
it is passive if $y^Tu = \varphi^T(u)u \geq 0$ for all $u\in \R^p$.
The system (\ref{eq:0.0}) is also said to be output feedback passive with index $\varepsilon > 0$ 
if (\ref{eq:0.1}) is replaced by
\begin{eqnarray}
S(x(t)) - S(x(0)) \leq \int^{t}_0 y^T(\tau)u(\tau) - \varepsilon \|y(\tau)\|^2 d\tau
\label{eq:0.3}
\end{eqnarray}
If the right-hand side of (\ref{eq:0.3}) is changed as
\begin{eqnarray}
S(x(t)) - S(x(0)) \leq \int^{t}_0 y^T(\tau)u(\tau) + \varepsilon \|u(\tau)\|^2 d\tau
\label{eq:0.4}
\end{eqnarray}
with $\varepsilon > 0$, the system is said to be passivity-short, and then $\varepsilon$
is called impact coefficient \cite{Qu}.
\end{definition}


\subsection{System Description}

In this paper, we consider a building with multiple zones $i = 1,2,\dots, n$.
The zones $i=1,2,\dots, n$ are divided into two groups:
The first group consists of zones equipped with VAV (Variable Air Volume) HVAC systems whose thermal dynamics is assumed to be 
modeled by the RC circuit model \cite{MPC1}--\cite{MPC4}, \cite{Xuan} as 
\begin{eqnarray}
C_i \dot T_i = \frac{T^{\rm a} - T_i}{R_i} + \sum_{j\in \N_i}\frac{T_j-T_i}{R_{ij}}+a_{i}(T^{\rm s}_i - T_i)m_i+q_i,
\label{eqn:1.1}
\end{eqnarray}
where $T_i$ is the temperature of zone $i$, $m_i$ is the mass flow rate
at zone $i$, $T^{\rm a}$ is the ambient temperature,
$T^{\rm s}_i$ is the air temperature supplied to zone $i$, 
which is treated as a constant throughout this paper,
$q_i$ is the heat gain at zone $i$ from external sources like occupants,
$C_i$ is the thermal capacitance, $R_i$ is the thermal resistance of
the wall/window, $R_{ij}$ is the thermal resistance between
zone $i$ and $j$ and $a_i$ is the specific heat of the air.

The second group is composed of other spaces such as walls and windows
whose dynamics is modeled as
\begin{eqnarray}
C_i \dot T_i = \frac{T^{\rm a} - T_i}{R_i} + \sum_{j\in \N_i}\frac{T_j-T_i}{R_{ij}}.
\label{eqn:1.2}
\end{eqnarray}
Rooms not in use can be categorized into this group.
Without loss of generality, we assume that $i = 1,2,\dots, n_1$ 
belong to the first group and $i = n_1+1, \dots, n\ (n_2 := n - n_1)$ belong to the second.
Remark that the system parameters $C_i$, $R_i$, and $R_{ij}$ can be identified
using the toolbox in \cite{toolbox}.

The collective dynamics of  (\ref{eqn:1.1}) and  (\ref{eqn:1.2}) is described as
\begin{eqnarray}
C\dot T = R T^{\rm a}{\bf 1} - RT - LT + BG(T)m + Bq,
\label{eqn:1.6}
\end{eqnarray}
where $T$, $q$ and $m$ are collections of 
$T_i\ (i=1,2,\dots, n)$, $q_i\ (i=1,2,\dots, n_1)$,
and $m_i\ (i=1,2,\dots, n_1)$ respectively.
The matrices $C$ and $R$ are diagonal matrices with diagonal elements
$C_i\ (i=1,2,\dots, n)$ and $\frac{1}{R_i}\ (i=1,2,\dots, n)$, respectively.
The matrix $L$ describes the weighted graph Laplacian
with elements $\frac{1}{R_{ij}}$, $G(T)\in \R^{n_1\times n_1}$ is a block diagonal matrix with diagonal elements
equal to $a_i(T^{\rm s}_i -T_i)\ (i=1,2,\dots,n_1)$,
${\bf 1}$ is the $n$-dimensional real vector whose elements are all 1,
and $B = [I_{n_1}\ 0]^{\top}\in \R^{n\times n_1}$.

We next linearize the model
at around an equilibrium as 
\begin{eqnarray}
C\dot{\delta T} = R \delta T^{\rm a}{\bf 1} - R\delta T - L\delta T + B G(\bar T)\delta m - \bar U\delta T+ B\delta q,
\nonumber
\end{eqnarray}
where $\delta T$, $\delta m$, $\delta T^{\rm a}$ and $\delta q$ describe the errors from the equilibrium states and inputs and 
$\bar U \in \R^{n\times n}$
is a diagonal matrix whose diagonal elements are $\bar m_1,\dots, \bar m_{n_1},0,\dots, 0$, 
where $\bar m_i$ is the $i$-th element of the equilibrium input $\bar m$.

Using the variable transformations 
\begin{eqnarray}
x \!\!&\!\!:=\!\!&\!\! C^{1/2}\delta T,\ 
u := B^{\top}C^{-1/2}BG(\bar T)\delta m, 
\nonumber\\
w_{\rm a} \!\!&\!\!:=\!\!& C^{-1/2}R \delta T^{\rm a}{\bf 1},\ w_{\rm q} := B^{\top}C^{-1/2}B\delta q,
\label{eqn:vt}
\end{eqnarray}
(\ref{eqn:1.2}) is rewritten as
\begin{eqnarray}
\dot x  = -A x + B u + Bw_{\rm q} + w_{\rm a},\ x := \begin{bmatrix}
x_1\\
x_2
\end{bmatrix}
\label{eqn:1.4}
\end{eqnarray}
where $A := C^{-1/2}(R+L + \bar U)C^{-1/2}$, $x_1 \in \R^{n_1}$ and $x_2 \in \R^{n_2}$.
Remark that the matrix $A$ is positive definite \cite{Hong}.

From control engineering point of view, $x$ is the system state, 
$u$ is the control input, and $w_{\rm q}$ and $w_{\rm a}$ are 
disturbances.
We suppose that $w_{\rm a}$ is measurable as well as $x_1$.
Meanwhile, it is in general hard to measure the heat gain $w_{\rm q}$.




\subsection{Optimization Problem}

Regarding the above system, we formulate the optimization problem to be solved as follows:
\begin{subequations}
\label{eqn:3.1}
\begin{eqnarray}
\hspace{-1cm}
&&\min_{z=[z_x^{\top}\ z_u^{\top}]^{\top}\in \R^{n+n_1}}  \|z_{x1} - h\|^2 + f(z_{u}) 
\label{eqn:3.1a}\\
\hspace{-1cm}
&&\mbox{subject to: }  g (z_u)  \leq 0
\label{eqn:3.1b}\\
\hspace{-1cm}
&&-Az_{x} + B z_{u} + Bd_{\rm q} + d_{\rm a} = 0
\label{eqn:3.1c}
\end{eqnarray}
\end{subequations}
The variables $z_x := [z_{x1}^{\top}\ z_{x2}^{\top}]^{\top}\ (z_{x1} \in \R^{n_1}, z_{x2} \in \R^{n_2})$, 
and $z_u\in \R^{n_1}$ correspond to the zone temperature 
$x$ and mass flow rate $u$ after the transformation (\ref{eqn:vt}).
The parameters $d_{\rm q}\in \R^{n_1}$ and $d_{\rm a}\in \R^n$
are DC components of the disturbances $w_{\rm q}$ and $w_{\rm a}$, respectively.
These variables are coupled by (\ref{eqn:3.1c}) which describes the stationary equation of (\ref{eqn:1.4}).

Throughout this paper, we assume the following assumption.

\begin{assumption}
\label{ass:0}
The problem (\ref{eqn:3.1}) satisfies the following properties:
(i) $f: \R^{n_1}\to \R$ is convex and its gradient is locally Lipschitz,
(ii) every element of the constraint function $g: \R^{n_1}\to \R^{c}$ is
convex and its gradient is locally Lipschitz, and 
(iii) there exists $z_u\in \R^{n_1}$ such that $g(z_u) < 0$.
\end{assumption}

The first term of (\ref{eqn:3.1a}) evaluates
the human comfort, where $h\in \R^{n_1}$ is the collection of the most comfortable temperatures for occupants in each zone,
which might be determined directly by occupants in the same way as the current systems, or computed 
using human comfort metrics like PMV (Predicted Mean Vote).
The quadratic function for the error 
is commonly employed in the MPC papers
\cite{MPC2,64}--\cite{72} and \cite{Xuan}.
Note that it is common to put weights on each element of $z_{x1} - h$
to give priority to each zone, but this can be done by appropriately scaling
each element of $z_u$ in the function $f$.
This is why we take (\ref{eqn:3.1a}).

The function $f$ is introduced to reduce power consumption.
The papers \cite{MPC2,64}--\cite{72} simply take a linear or quadratic function
of control efforts as such a function and then Assumption \ref{ass:0}(i) is trivially satisfied.
Also, as mentioned in \cite{ma,Xuan}, the power consumption of supply fans
is approximated by the cube of the sum of the mass flow rates, which
also satisfies Assumption \ref{ass:0}(i).
A simple model of consumption at the cooling coil is given by 
the product of the mass flow rate $m_i$ and $|T_i^{\rm s} - T^{\rm a}|$ \cite{BE14}, which also belongs to the intended class
\footnote{For simplicity, we skip the dependence of $d_{\rm a}$ on $f$, but
subsequent results are easily extended to the case that $f$ depends on $d_{\rm a}$.}.
The constraint function $g: \R^{n_1}\to \R^c$ reflects hardware constraints
and/or an upper bound of the power consumption.
For example, the constraints in \cite{Xuan} are reduced to the above form.

The objective of this paper is to design a controller so as to
ensure convergence of the actual room temperature $x_1$
to the optimal room temperature $z_{x1}^*$, the solution to (\ref{eqn:3.1}),
without direct measurements of $w_{\rm q}$.

\section{Optimization Dynamics}
\label{sec:3}

\subsection{Optimization Dynamics}

In this subsection, we present a dynamics to solve the above optimization problem.
Before that, we eliminate $z_x$ from (\ref{eqn:3.1}) using (\ref{eqn:3.1c}).
Then, the problem is rewritten as
\begin{subequations}
\label{eq:3.3}
\begin{align}
&\min_{z_u\in \R^{n_1}}  \|
B^{\top}A^{-1}(Bz_{u} + Bd_{\rm q} + d_{\rm a} - \bar h)\|^2 + f(z_{u}) 
\label{eqn:3.3a}\\
&\mbox{subject to: }  g (z_u)  \leq 0 
\label{eqn:3.3b}
\end{align}
\end{subequations}
with $\bar h = ABh$.
It is easy to confirm from positive definiteness of $A$ that 
the cost function of (\ref{eq:3.3}) is strongly convex.
In this case, (\ref{eq:3.3}) has the unique optimal solution, denoted by $z^*_u$, and it
satisfies the following KKT conditions \cite{boyd}.
\begin{subequations}
\label{eqn:3.5}
\begin{align}
&M^2(z_u^* + d_{\rm q})  + N(d_{\rm a} - \bar h) 
+ \nabla f(z_u^*) + \nabla g(z_u^*) \lambda^* = 0,
\label{eqn:3.5a}\\
&g( z_u^*) \leq 0,\ \lambda^* \geq 0,\ \lambda^*\circ g(z_u^*) = 0,
\label{eqn:3.5b}
\end{align}
\end{subequations}
where $M = B^{\top}A^{-1}B$, $N := B^{\top}A^{-1}BB^{\top}A^{-1}$, and $\circ$ represents the Hadamard product.
Since  (\ref{eq:3.3}) is essentially equivalent to (\ref{eqn:3.1}), the solution to (\ref{eq:3.3}) also provides
a solution for the problem (\ref{eqn:3.1}).
Precisely, if we define $z_x^* \in \R^n$ as
$z_x^* := A^{-1}(Bz_u^*+Bd_{\rm q}+d_{\rm a})$,
the pair $z^* := [(z_x^*)^{\top}\ (z_u^*)^{\top}]^{\top}$ is a solution to  (\ref{eqn:3.1}).
In the sequel, we also use the notation $z_x^* := [(z_{x1}^*)^{\top}\ (z_{x2}^*)^{\top}]^{\top}$ ($z_{x1}^*\in \R^{n_1}$ and $z_{x2}^*\in \R^{n_2}$).
It is then easy to confirm
\begin{align}
z^*_{x1} = M(z_u^* + d_{\rm q}) + B^{\top}A^{-1}d_{\rm a}.
\label{def:zx1}
\end{align}

Given $d_{\rm q}$ and $d_{\rm a}$, it would be easy to solve (\ref{eqn:3.5}).
However, in the practical applications, 
it is desired that $d_{\rm q}$ and $d_{\rm a}$ are updated in real time according to
the changes of disturbances.
In this regard, it is convenient to take a dynamic solution process of optimization
since it trivially allows one to update the parameters in real time.
In particular, we employ the primal-dual gradient algorithm \cite{cortes}
as one of such solutions\footnote{Another benefit of using the dynamic solution is that it provides
a distributed solution when the present results are extended to a more global problem, 
although it exceeds the scope of this paper. Please refer to \cite{CDC} for more details on the issue.}.
However, it is hard to obtain $d_{\rm q}$ since $w_{\rm q}$ is not measurable. 
We thus need to estimate $d_{\rm q}$ from the measurements of physical quantities.
This motivates us to interconnect the physical dynamics with the optimization dynamics.

Taking account of the above issues, we present 
\begin{subequations}
\label{eqn:3.4}
\begin{eqnarray}
\!\!\!\!\!\!\!\!\!\!\dot{\hat{z}}_u \!\!&\!\!=\!\!&\!\! -\alpha \{M^2(\hat z_u + \hat d_{\rm q}) + 
N(w_{\rm a} - \bar h) + \nabla f(\hat z_u) + p\},
\label{eqn:3.4a}\\
\!\!\!\!\!\!\!\!\!\! \dot{\hat{\lambda}} \!\!&\!\!=\!\!&\!\! [g(\hat z_u)]^+_{\hat \lambda},\ p = \nabla g(\hat z_u)\hat \lambda,
\label{eqn:3.4b}
\end{eqnarray}
\end{subequations}
where $\hat z_u$ and $\hat \lambda$ are estimates of $z_u^*$ and $\lambda^*$ respectively,
and $\alpha > 0$. 
The notation $[b]^+_a$ for real vectors $a,b$ with the same dimension 
provides
a vector whose $l$-th element, denoted by $([b]^+_a)_l$, is given by
\begin{eqnarray}
([b]^+_a)_l = \left\{
\begin{array}{l}
0,\mbox{ if } a_l = 0 \mbox{ and }b_l< 0\\
b_l,\mbox{ otherwise}
\end{array}
\right.,
\label{eqn:3.6}
\end{eqnarray}
where $a_l, b_l$ are the $l$-th element of $a,b$, respectively.
Note that (\ref{eqn:3.4}) is different from the primal-dual gradient algorithm for (\ref{eqn:3.5})
in that the term $d_{\rm a}$ is replaced by the measurement $w_{\rm a}$, and
$d_{\rm q}$ is replaced by its estimate $\hat d_{\rm q}$ whose production will be mentioned later.
The system is illustrated in Fig. \ref{fig:1}.


\begin{figure}\centering\centering
\includegraphics[width=7.5cm]{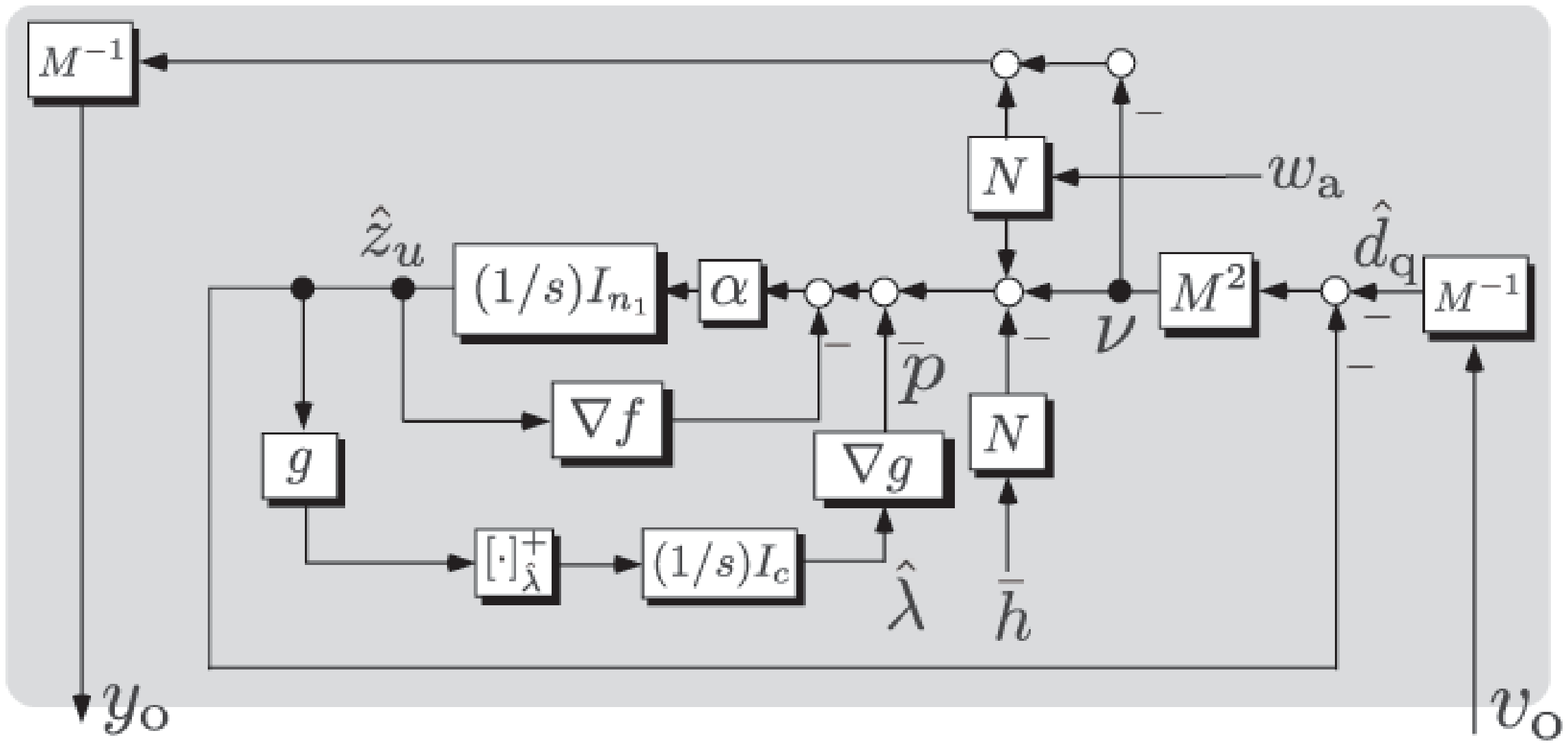}
\caption{Block diagram of optimization dynamics,
which is passive
from $\tilde v_{\rm o} = v_{\rm o} - v_{\rm o}^*$ 
$\tilde y_{\rm o} = \hat y_{\rm o} - y_{\rm o}^*$ with $v_{\rm o}^* := Md_{\rm q}$ and
$y_{\rm o}^* := z_{x1}^*$ 
(Lemma \ref{lem:hata7}).}
\label{fig:1}
\end{figure}

\subsection{Passivity Analysis for Optimization Dynamics}
\label{sec:3.2}

Hereafter, we analyze passivity of the above optimization process 
assuming that $w_{\rm a}$ is constant.
In this case, $w_{\rm a} \equiv d_{\rm a}$ holds. 
In practice, the disturbance $w_{\rm a}$, namely the ambient temperature $T^{\rm a}$,
is time-varying but the following results are applied to the practical case
if $w_{\rm a}$ is approximated by a piecewise constant signal,
which is fully expected since $T^{\rm a}$ usually varies slowly.

Under the above assumption, we define the output $\nu := - M^2(\hat z_u + \hat d_{\rm q})$ for (\ref{eqn:3.4}).
We then have the following lemma.
\begin{lemma}
\label{lem:3}
Consider the system (\ref{eqn:3.4}) with $w_{\rm a} \equiv d_{\rm a}$ and $\hat \lambda(0) \geq 0$. 
Then, under Assumption \ref{ass:0}, it is passive from $\tilde d_{\rm q} := \hat d_{\rm q} - d_{\rm q}$ to $-\tilde \nu$,
where $\tilde \nu := \nu - \nu^*$ and $\nu^* := -M^2(z^*_u + d_{\rm q})$.
\end{lemma}

\begin{proof}
See Appendix \ref{appendix:1}.
\end{proof}

Let us next transform the output $\nu$ to
\begin{align}
y_{\rm o} &:= -M^{-1}\nu + B^{\top}A^{-1} d_{\rm a} = M(\hat z_u + \hat d_{\rm q}) + B^{\top}A^{-1} d_{\rm a}, 
\nonumber\\
y_{\rm o} ^* &:= -M^{-1}\nu^* + B^{\top}A^{-1} d_{\rm a}.
\nonumber
\end{align}
Comparing (\ref{def:zx1}) and the above definition of $y_{\rm o}$, the signal
$y_{\rm o}$ is regarded as an estimate of $z^*_{x1}$.
We also define $v_{\rm o} := M\hat d_{\rm q}$ and $v_{\rm o}^* := Md_{\rm q}$.
Then, we can prove the following lemma.

\begin{lemma}
\label{lem:hata7}
Consider the system (\ref{eqn:3.4})
with $w_{\rm a} \equiv d_{\rm a}$ and $\hat \lambda(0) \geq 0$. 
Then,  under Assumption \ref{ass:0}, it is output feedback passive from $\tilde v_{\rm o} := v_{\rm o} - v_{\rm o}^*$ to 
$\tilde y_{\rm o} := y_{\rm o} - y_{\rm o}^*$ with index $1$.
\end{lemma}

\begin{proof}
It is easy to see from $\nu^* = -M^2(z_u^* + d_{\rm q})$ 
and (\ref{def:zx1}) that
$y_{\rm o}^* = z^*_{x1}$.
From (\ref{eqn:3.18}), $\tilde y_{\rm o} = -M^{-1}\tilde \nu$ and $\tilde v_{\rm o} = M\tilde d_{\rm q}$,
we have the following inequality.
\begin{eqnarray}
D^+S_{\rm o} \leq 
\tilde y_{\rm o}^{\top}\tilde v_{\rm o} - \|\tilde y_{\rm o}\|^2
\label{eqn:3.197}
\end{eqnarray}
Integrating this in time completes the proof.
\end{proof}



\section{Physical Dynamics}

In this section, we design a physical dynamics and prove its passivity.
A passivity-based design for the model  (\ref{eqn:1.4}) is 
presented in \cite{jtwen}, but we modify the control architecture
in order to interconnect it with the optimization dynamics.

\subsection{Controller Design}

In this subsection, we design a controller to determine the input $u$ so that
$x_1$ tracks a reference signal $r$.
Here we assume the following assumption, where
\[
A = \begin{bmatrix}
A_1&A_2^{\top}\\
A_2&A_3
\end{bmatrix},\ A_1 \in \R^{n_1\times n_1},\
A_3 \in \R^{n_2\times n_2}.
\]
\begin{assumption}
\label{ass:1}
The matrix $MA_1+A_1M$ is positive definite.
\end{assumption}
This property does not always hold for any positive definite matrices $A_1$ and $M$, 
but it is expected to be true in many practical cases since
the diagonal elements tend to be dominant both for $A_1$ and 
$M = (A_1 - A_2^{\top}A_3^{-1}A_2)^{-1}$ in this application \cite{Xuan}. 
Note that this assumption holds in the full actuation case ($A_1 = A, M = A^{-1}$).
Inspired by the fact that many existing systems employ
Proportional-Integral (PI) controllers (with logics) as the local controller,
we design the following controller adding reference and disturbance feedforward terms.
\begin{subequations}
\label{eqn:2.8}
\begin{eqnarray}
\dot \xi \!\!&\!\!=\!\!&\!\! k_{\rm I}(r - x_1)
\label{eqn:2.8a}\\
u \!\!&\!\!=\!\!&\!\! k_{\rm P}(r - x_1) + \xi + \kappa r + F w_{\rm a}
\label{eqn:2.8b}
\end{eqnarray}
\end{subequations}
where $k_{\rm P} > 0,\ k_{\rm I} > 0$ and $F := [-I_{n_1}\ A_2^{\top}A_3^{-1}]$ and
$I_{n_1}$ is the $n_1$-by-$n_1$ identity matrix.
The feedforward gain $\kappa > 0$ is selected so that
\begin{eqnarray}
P :=  MA_1+A_1M - 2\kappa M > 0.
\label{eqn:9.1}
\end{eqnarray}
Such a $\kappa$ exists under Assumption \ref{ass:1}.

%


Substituting (\ref{eqn:2.8}) into (\ref{eqn:1.4}) yields
\begin{subequations}
\label{eqn:2.9}
\begin{eqnarray}
\dot x \!\!&\!\!=\!\!&\!\! -Ax
+ k_{\rm P} B (r  - x_1) + \kappa B  r  
+ B\xi 
\nonumber\\
&&\hspace{1cm}+ Bw_{\rm q} + (BF+I_n) w_{\rm a},
\label{eqn:2.9a}\\
\dot \xi \!\!&\!\!=\!\!&\!\! k_{\rm I}(r - x_1).
\label{eqn:2.9b}
\end{eqnarray}
\end{subequations}
For the system, we have the following lemma.

\begin{lemma}
\label{lem:steady}
The steady states $x^*$ 
and $\xi^*$ of (\ref{eqn:2.9}) 
for $r \equiv r^*$,
$w_{\rm a} \equiv d_{\rm a}$ and $w_{\rm q} \equiv d_{\rm q}$
are given as follows.
\begin{eqnarray}
x^* \!\!&\!\!=\!\!&\!\! -F^{\top} r^* + \begin{bmatrix}
0\\
A_3^{-1}B_c^{\top}
\end{bmatrix}d_{\rm a},\ 
B_c := 
\begin{bmatrix}
0\\
I_{n_2}
\end{bmatrix}\in \R^{n\times n_2}
\nonumber\\
\xi^{*} \!\!&\!\!=\!\!&\!\! (M^{-1} - \kappa I_{n_2})r^* - d_{\rm q}.
\label{eqn:99.3}
\end{eqnarray}
\end{lemma}


Equation (\ref{eqn:2.9}) is now rewritten as
\begin{subequations}
\label{eqn:99.5}
\begin{eqnarray}
\dot x \!\!&\!\!=\!\!&\!\!  
-\bar A x + BM^{-1}\zeta
+ Bw_{\rm q} + \bar F w_{\rm a}
\label{eqn:99.5a}\\
\dot \xi \!\!&\!\!=\!\!&\!\! k_{\rm I}(r - x_1),\
\zeta = \bar k_{\rm P} M (r  - x_1) 
+ M\xi,
\label{eqn:99.5b}
\end{eqnarray}
\end{subequations}
where
\[
\bar A := A - \kappa BB^{\top},\ 
\bar k_{\rm P} := k_{\rm P}+\kappa,\ \bar F := \begin{bmatrix}
-A_2^{\top}A_3^{-1}\\
I_{n_2}
\end{bmatrix}B_c^{\top}.
\]
Remark that, at the steady state, the variable $\zeta$ is
equal to 
\begin{eqnarray}
\zeta^* := M\xi^{*} = K r^* - Md_{\rm q},\ \ 
K := I_{n_1} - \kappa M.
\label{eqn:99.6}
\end{eqnarray}

\subsection{Passivity Analysis for Physical Dynamics}

\begin{figure}\centering\centering
\includegraphics[width=8.4cm]{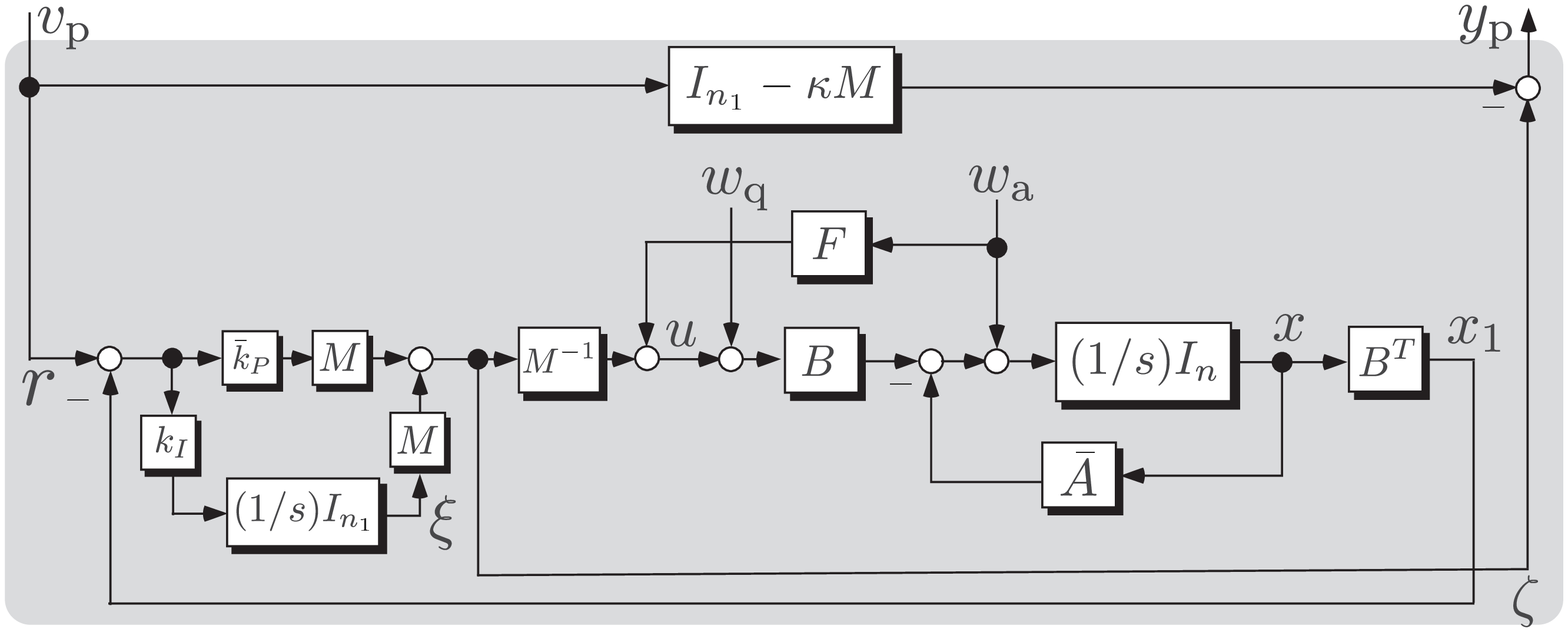}
\caption{Block diagram of the physical dynamics with the local controller, which 
is passivity-short from $\tilde v_{\rm p} = v_{\rm p} - v_{\rm p}^*$ to 
$\tilde y_{\rm p} = y_{\rm p} - y_{\rm p}^*$ with $v_{\rm p}^* := z_{x1}^*$ and $y_{\rm p}^* := -Md_{\rm q}$ (Lemma \ref{lem:hata5}).}
\label{fig:4}
\end{figure}

In this subsection, we analyze passivity of the system (\ref{eqn:99.5})
 assuming that $w_{\rm q}$ and $w_{\rm a}$ are constant.
The case of the time varying $w_{\rm q}$ will be treated in the end of the next section.

Choose $\zeta$ as the output and prove
passivity as follows. 

\begin{lemma}
\label{lem:hata12}
Consider the system (\ref{eqn:99.5}) with 
$r \equiv r^*$, $w_{\rm a} \equiv d_{\rm a}$ and $w_{\rm q} \equiv d_{\rm q}$.
Then, under Assumption \ref{ass:1}, the system is 
passive from $\tilde r := r - r^*$
to $\tilde \zeta := \zeta - \zeta^*$.
\end{lemma}
\begin{proof}
See Appendix \ref{appendix:2}.
\end{proof}
Remark that this lemma holds regardless of the value of $r^*$.

To extract the term $Md_{\rm q}$ from (\ref{eqn:99.6}), we define the output 
\begin{eqnarray}
y_{\rm p} := \zeta - K r,\ y_{\rm p}^* := \zeta^* - K r^* = - Md_{\rm q},
\label{eqn:9.15}
\end{eqnarray}
and $v_{\rm p} := r$ and $v_{\rm p}^* := z_{x1}^*$.
Then, we have the following.
\begin{lemma}
\label{lem:hata5}
Consider the system (\ref{eqn:99.5}) with 
$r \equiv r^*$,
$w_{\rm a} \equiv d_{\rm a}$ and $w_{\rm q} \equiv d_{\rm q}$.
Then, under Assumption \ref{ass:1}, the system 
from $\tilde v_{\rm p} := v_{\rm p} - v_{\rm p}^*$ to $\tilde y_{\rm p} := y_{\rm p} - y_{\rm p}^*$
is  passivity-short with the impact coefficient
$1 - \kappa\sigma$, where $\sigma > 0$ is the minimal eigenvalue of $M$.
\end{lemma}
\begin{proof}
If we take $r^* = z_{x1}^*$, we have
$\tilde y_{\rm p} = \tilde \zeta - K \tilde v_{\rm p}$.
Substituting this into (\ref{eqn:9.14}) yields
\begin{align}
\dot S_{\rm p}&\leq
\tilde y_{\rm p}^{\top} \tilde v_{\rm p} +
\tilde v_{\rm p}^{\top} K\tilde v_{\rm p}
- k_{\rm P}(\tilde v_{\rm p} - \tilde x_1)^{\top}M(\tilde v_{\rm p} - \tilde x_1)
\label{eqn:not_sigma}\\
&\leq
\tilde y_{\rm p}^{\top} \tilde v_{\rm p} +
(1-\kappa\sigma)\|\tilde v_{\rm p}\|^2
-k_{\rm P}\sigma \|\tilde v_{\rm p} - \tilde x_1\|^2
\label{eqn:9.15}
\end{align}
This completes the proof.
\end{proof}

\section{Interconnection of Optimization and Physical Dynamics}
\label{sec:5}

Let us interconnect the optimization dynamics
(\ref{eqn:3.4}) and physical dynamics (\ref{eqn:99.5}). 
Remark that the stationary value of $y_{\rm o}$, $y_{\rm o}^* = z_{x1}^*$,
is equivalent to that of $v_{\rm p}$, $v_{\rm p}^* = z_{x1}^*$.
Also, $v_{\rm o}^* = -y_{\rm p}^*$ holds.
Inspired by these facts, we interconnect these systems via the 
negative feedback as
$v_{\rm o} = - y_{\rm p},\  v_{\rm p} = y_{\rm o}$.
We then have the following main result of this paper.


\begin{theorem}
\label{thm:1}
Suppose that $\hat \lambda(0)\geq 0$, 
$w_{\rm a} \equiv d_{\rm a}$ and $w_{\rm q} \equiv d_{\rm q}$.
Then, if Assumptions \ref{ass:0} and \ref{ass:1} hold,
the interconnection of (\ref{eqn:3.4}) and 
(\ref{eqn:99.5}) via $v_{\rm o} = - y_{\rm p},\  v_{\rm p} = y_{\rm o}$
ensures that $x_1 \to z_{x1}^*$.
\end{theorem}

\begin{proof}
Define $S := S_{\rm o} + S_{\rm p}$.
Then, combining (\ref{eqn:3.197}), (\ref{eqn:9.15}) and $v_{\rm o} = - y_{\rm p},\  v_{\rm p} = y_{\rm o}$ yields
\begin{eqnarray}
D^+S \!\!&\!\!\leq\!\!&\!\! 
-\kappa\sigma\|\tilde y_{\rm o}\|^2
-k_{\rm P} \sigma \|\tilde v_{\rm p} - \tilde x_1\|^2.
\label{eqn:9.16}
\end{eqnarray}
This means that both of $\tilde y_{\rm o} = y_{\rm o} - z_{x1}^*$ and 
$\tilde y_{\rm o} - \tilde x_1 = y_{\rm o} - x_1$ belong to class ${\mathcal L_2}$.
Since $S$ is positive definite, all of the state variables
$\hat z_u, \hat \lambda, x$ and $\xi$ belong to ${\mathcal L}_{\infty}$.
From (\ref{eqn:3.4}), $\dot{\tilde{y}}_{\rm o} = -M^{-1}\dot \nu = M\dot{\hat{z}}_u$ is bounded.
Also, (\ref{eqn:2.9}) means that $\dot{\tilde{x}} = \dot x$ is bounded and hence
$\dot y_{\rm o} - \dot x_1$ is bounded.
Thus, invoking Barbalat's lemma, we can prove
$y_{\rm o} - z_{x1}^* \to 0,\ \ y_{\rm o} - x_1 \to 0$,
which means $x_1 \to z_{x1}^*$.
This completes the proof.
\end{proof}
Lyapunov stability of the desirable equilibrium, tuple of $z_u^*, \lambda^*, x^*$ and $\xi^*$,
is also proved in the above proof.

It is to be emphasized that the optimal solution is 
dependent on the unmeasurable disturbance.
Nevertheless, convergence to the solution is guaranteed owing to
the feedback path from physics to optimization.


The above results are obtained assuming that both of $w_{\rm q}$ and $w_{\rm a}$ are constant.
This is likely valid for $w_{\rm a}$ since the ambient temperature is in general
slowly varying.
However, the heat gain $w_{\rm q}$ may contain high frequency components.
To address the issue, we decompose the signal $w_{\rm q}$ into 
the DC components $d_{\rm q}$ and others $\tilde w_{\rm q}$
as $w_{\rm q} = d_{\rm q} + \tilde w_{\rm q}$.
We also assume that $\tilde w_{\rm q}$ belongs to an extended ${\mathcal L}_2$ space \cite{BAW_BK}.
The following corollary then holds, which is proved following
the proof procedure of the well-known passivity theorem
\cite{BAW_BK,hatanaka} and using the fact that the right-hand side of
(\ref{eqn:9.16}) is upper bounded by
$-\frac{\kappa k_{\rm P}\sigma}{\kappa + k_{\rm P}} \|\tilde x_1\|^2$.
\begin{corollary}
\label{cor:1}
Suppose that $\hat \lambda(0)\geq 0$,
$w_{\rm a} \equiv d_{\rm a}$ and $w_{\rm q} = d_{\rm q} + \tilde w_{\rm q}$.
Then, if Assumptions \ref{ass:0} and \ref{ass:1} hold,
the interconnected system (\ref{eqn:3.4}),
(\ref{eqn:99.5}) and $v_{\rm o} = - y_{\rm p},\  v_{\rm p} = y_{\rm o}$ from $\tilde w_{\rm q}$
to $\tilde x_1 = B^{\top}\tilde x = x_1 - z^*_{x1}$ has a finite ${\mathcal L}_2$ gain.
\end{corollary}

We give some remarks on the present architecture.

Figs. \ref{fig:1} and \ref{fig:4} are oriented by theoretical analysis, 
but the implementation does not need to follow the information processing
in the figures.
Indeed, the interconnected system is equivalently transformed into Fig. \ref{fig:6}.
If we let the operations shaded by dark gray be executed
in the high-level controller,
the low-level controller can be 
implemented in a decentralized fashion similarly to the existing systems.


In Fig. \ref{fig:6}, both of the high-level and low-level controller
with the physical dynamics are biproper and hence a problem of algebraic loops can occur. 
This however does not matter in practice since the information transmissions
between high- and low-level processes usually suffer from possibly small delays.
Although the high-level controller itself contains an algebraic loop,
it is easily confirmed that the loop can be solved by direct calculations of the algebraic constraint.

\begin{figure}\centering
\includegraphics[width=7.5cm]{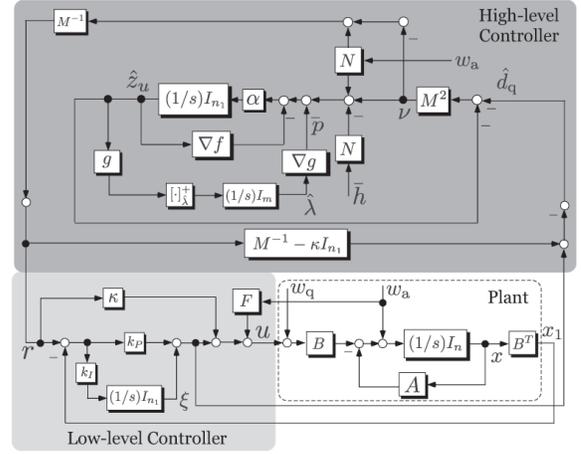}
\caption{Hierarchical control architecture.}
\label{fig:6}
\end{figure}

The transfer function from $w_{\rm q}$ to the disturbance estimate 
$\hat d_{\rm q}$, roughly speaking,
is almost the same as the complementary transfer function and
hence only the low frequency components are provided by the physical dynamics.
This is why $\hat d_{\rm q}$ is regarded 
as an estimate of the DC component of $w_{\rm q}$.
The cutoff frequency of the disturbance can be in principle tuned by $k_{\rm P}$ and $k_{\rm I}$,
but, once a closed-loop system is designed, the cutoff is also automatically decided.
It is however not always a drawback at least qualitatively. 
Actually, even if optimal solutions reflecting much faster disturbance variations are provided,
the physical states cannot respond to variations faster than the bandwidth.

It is a consequence of the internal model control and constant disturbances
that the actual disturbance $w_q$ is correctly estimated.
A control architecture based on a similar concept is presented in Section VII of
Stegink et al. \cite{SPV_16}.
However, it is clear that the problem (\ref{eqn:3.1}) does not meet the structural constraints
assumed in \cite{SPV_16} and hence
the architecture in \cite{SPV_16} cannot be directly applied to our problem.

Zhang et al. \cite{Xuan} present another kind of interconnection between physical and 
optimization dynamics based on a quasi-disturbance feedforward,
where the disturbance is computed by state measurements and their derivatives,
and then fed back to the optimization dynamics.
The differences of the present scheme from \cite{Xuan}
are listed as follows:

The approach of \cite{Xuan} requires the measurements of state variables.
If the states include temperatures of windows and/or walls,
its technological feasibility may be problematic or at least increases the system cost.
On the other hand, our approach needs only $x_1$ which is usually measurable.

Since there is no sensor to measure $\dot x$,
it has to be computed using the difference approximation,
which provides approximation errors.
Meanwhile, the present approach does not need such an approximation.

In \cite{Xuan}, the difference approximation errors together with sensor noises
and high frequency components of the disturbances 
are directly sent to the optimization dynamics, which may cause 
fluctuations for the output and internal variables in the optimization process
unless it is carefully designed in the sense of the noise reduction.
Adding a low-pass filter to the computed disturbance
might eliminate these undesirable factors.
However, the filter is not designed independently of 
 stability of the entire system in the presence of
uncertainties in $\dot x$ and the system model since, in this case, the quasi-feedforward system becomes a feedback system and the filter is included into the loop.
Meanwhile, the noises are automatically rejected by the physical dynamics in our algorithm.



\section{Simulation}


%
%

%

In this section, we demonstrate the presented control architecture 
through simulation.
For this purpose, we build a building on 3D modeling software
SketchUp (Trimble Inc.), which contains three rooms ($n_1 = 3$)
and other 46 zones ($n_2 = 46$) including walls, ceilings, and windows.
The building model is then installed into EnergyPlus \cite{ep}
in order to simulate the evolution of zone temperatures.
Then, using the acquired data, we identify 
the model parameters in (\ref{eqn:1.1}) and (\ref{eqn:1.2}) via BRCM toolbox \cite{toolbox}.

We next specify the optimization problem (\ref{eqn:3.1}). 
All the elements of $h$ are set to $22C^{1/2}{}^{\circ}$C and we take $f(z_u) = 150\|z_u\|^2$.
The constraints are also selected as
$|z_{ui}| \leq 0.61\ i = 1,2,3$ and $\sum_{i=1}^3|z_{ui}| \leq 1.25$,
where $z_{ui}$ is the $i$-th element of $z_u$.
Collecting these constraints, we define the function $g$.
However, since it turns out that directly using $g(z_u)\leq 0$ has a response speed problem in
penalizing the constraint violation in the primal-dual algorithm,
we instead take the constraint $\theta g(z_u) \leq 0$ with $\theta = 15$,
which does not essentially change the optimization problem.


In the simulation, we 
take the feedback gains $k_{\rm P} = 6.0 \times 10^{-2}$ and
$k_{\rm I} = 1.0 \times 10^{-3}$, and $\kappa = 1.0 \times 10^{-3}$,
which are tuned so that the peak gain of $\sigma$-plot from $r$ to $x_1$
is smaller than the well-known criterion.
It is then confirmed that Assumption \ref{ass:1} is satisfied.

\begin{figure}
\begin{minipage}[b]{4.2cm}
\includegraphics[width=4.2cm]{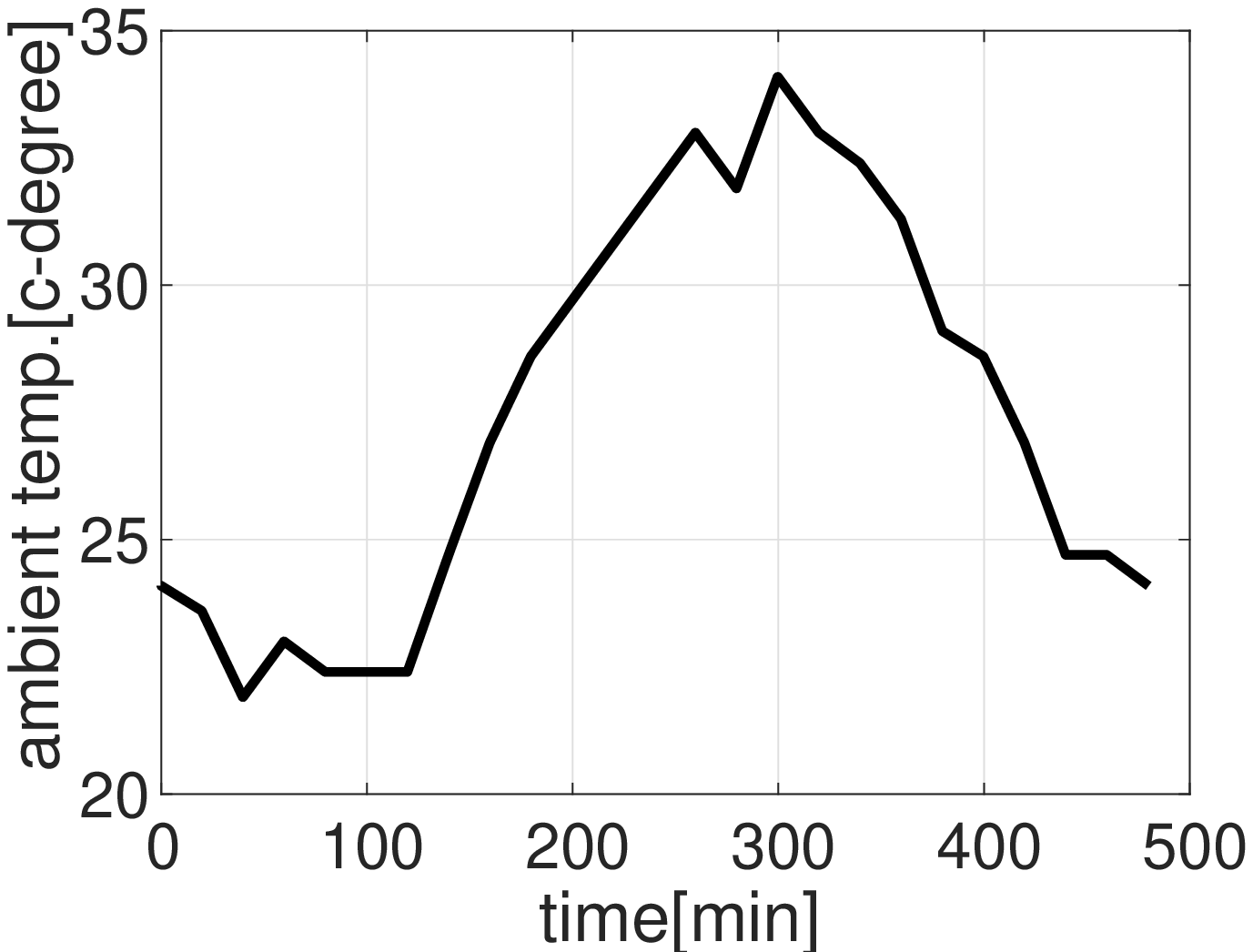}
\end{minipage}
\begin{minipage}[b]{4.2cm}
\includegraphics[width=4.2cm]{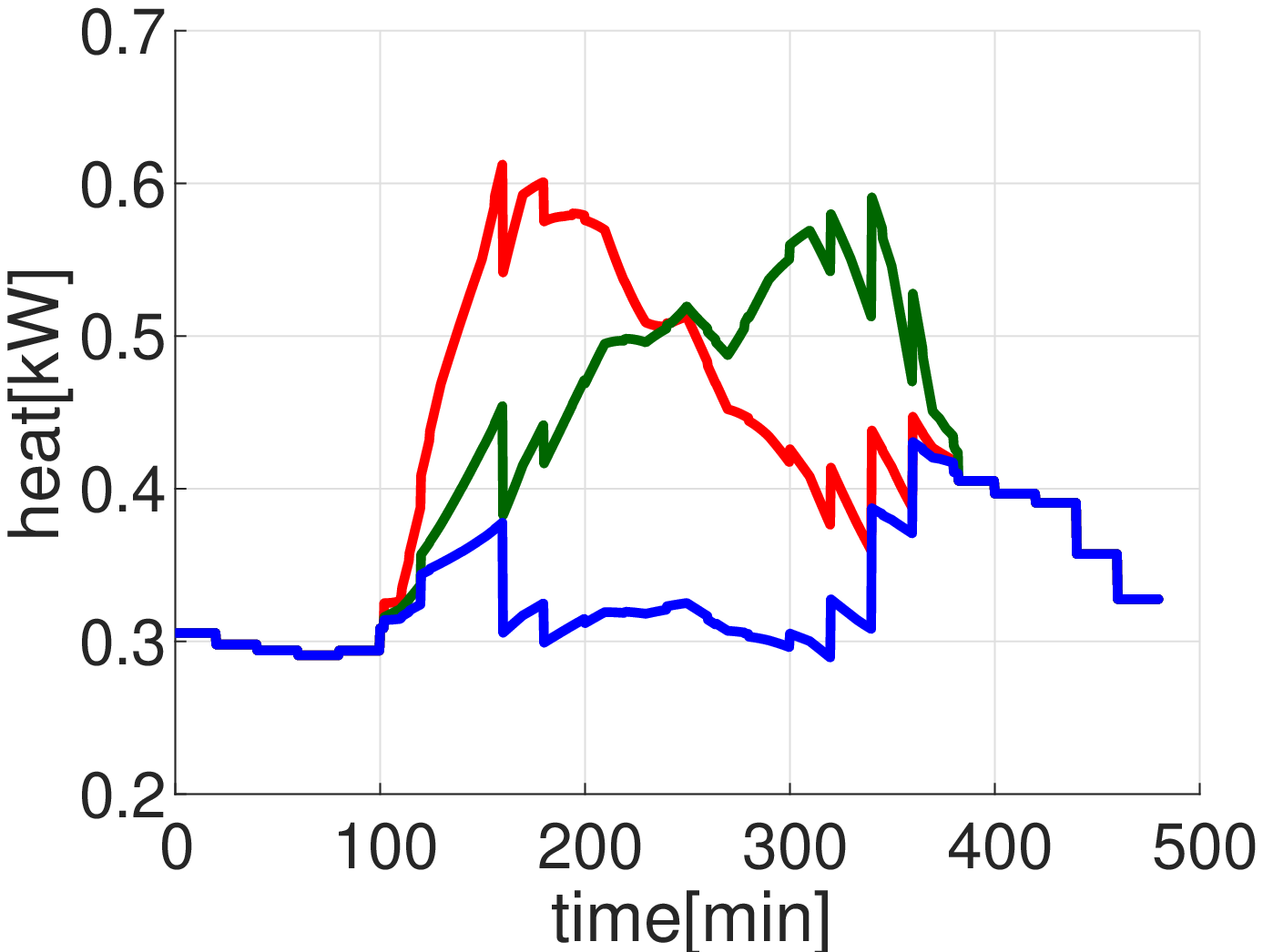}
\end{minipage}
\caption{Ambient temperature (left) and external heats (right).}
\label{fig:15}
\medskip

\begin{minipage}[b]{4.2cm}
\includegraphics[width=4.2cm]{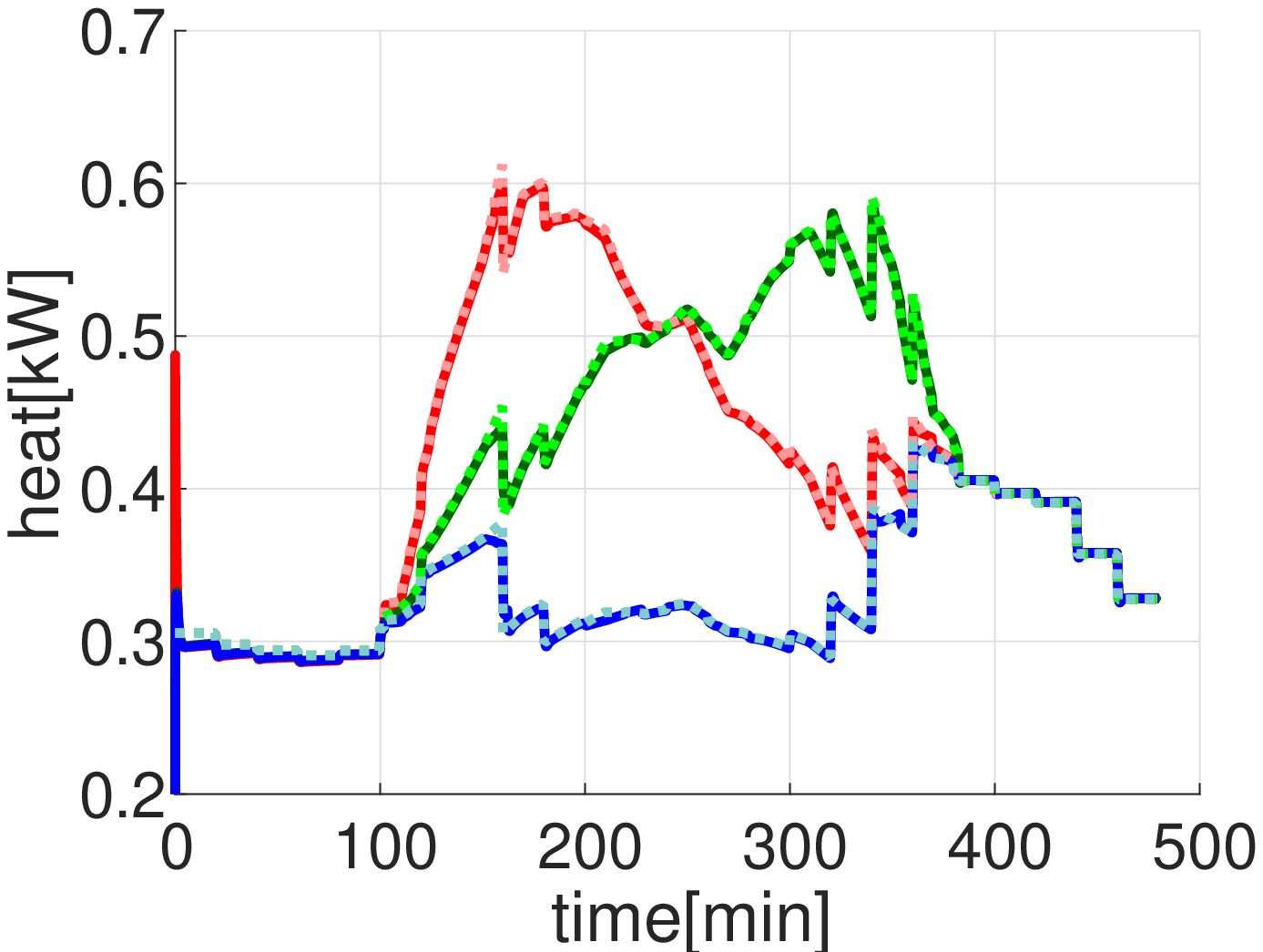}
\end{minipage}
\begin{minipage}[b]{4.2cm}
\includegraphics[width=4.2cm]{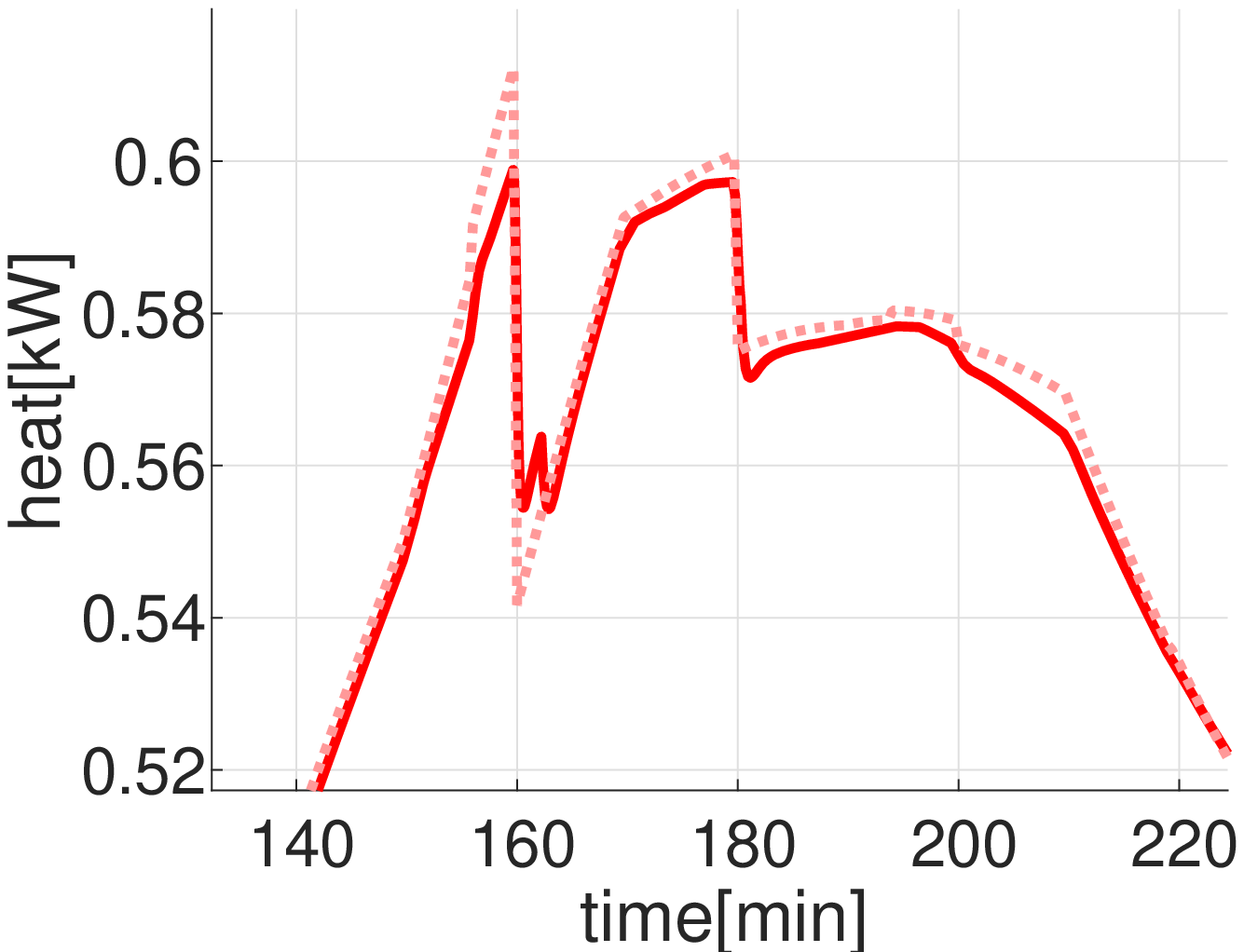}
\end{minipage}
\medskip

\begin{minipage}[b]{4.2cm}
\includegraphics[width=4.2cm]{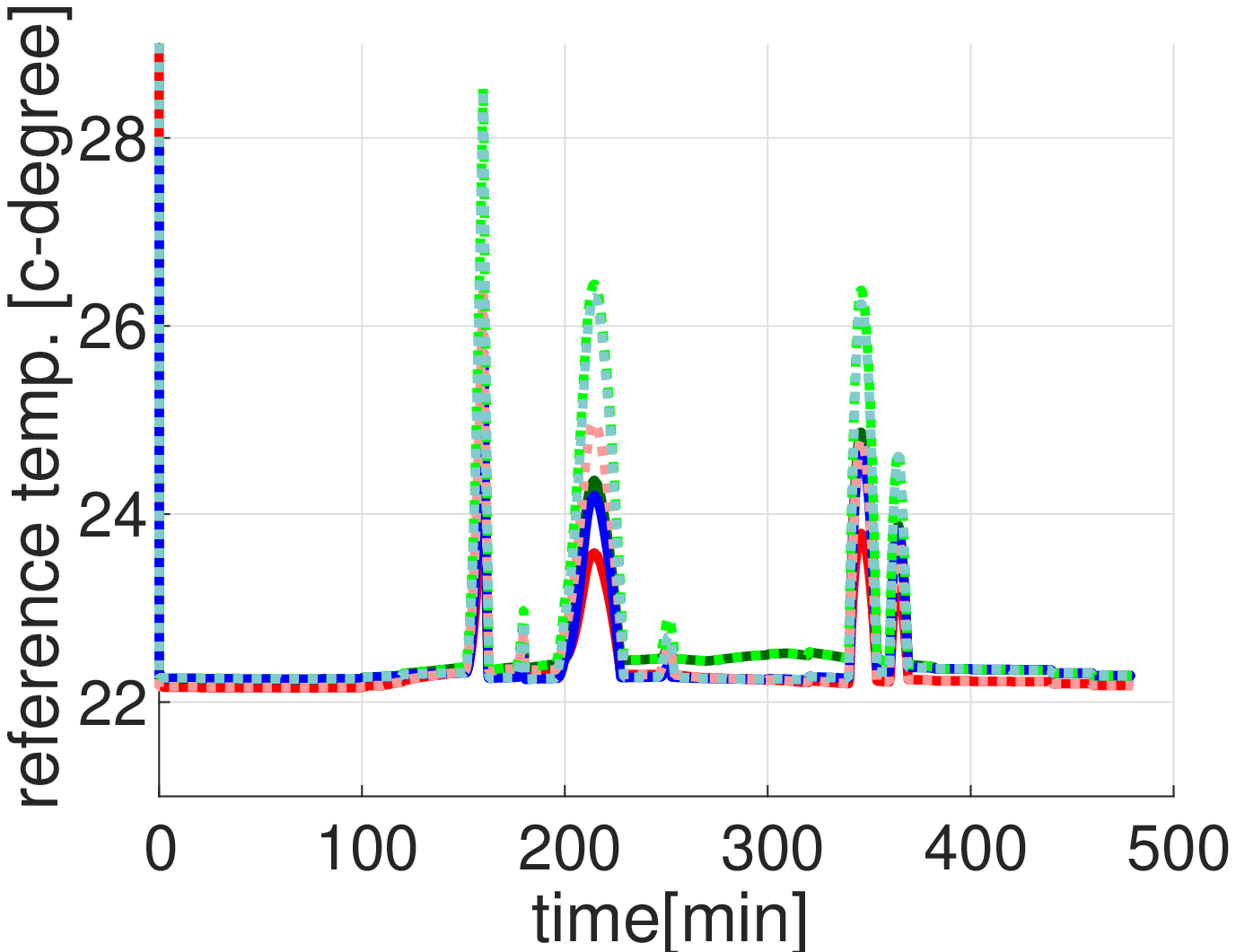}
\end{minipage}
\begin{minipage}[b]{4.2cm}
\includegraphics[width=4.2cm]{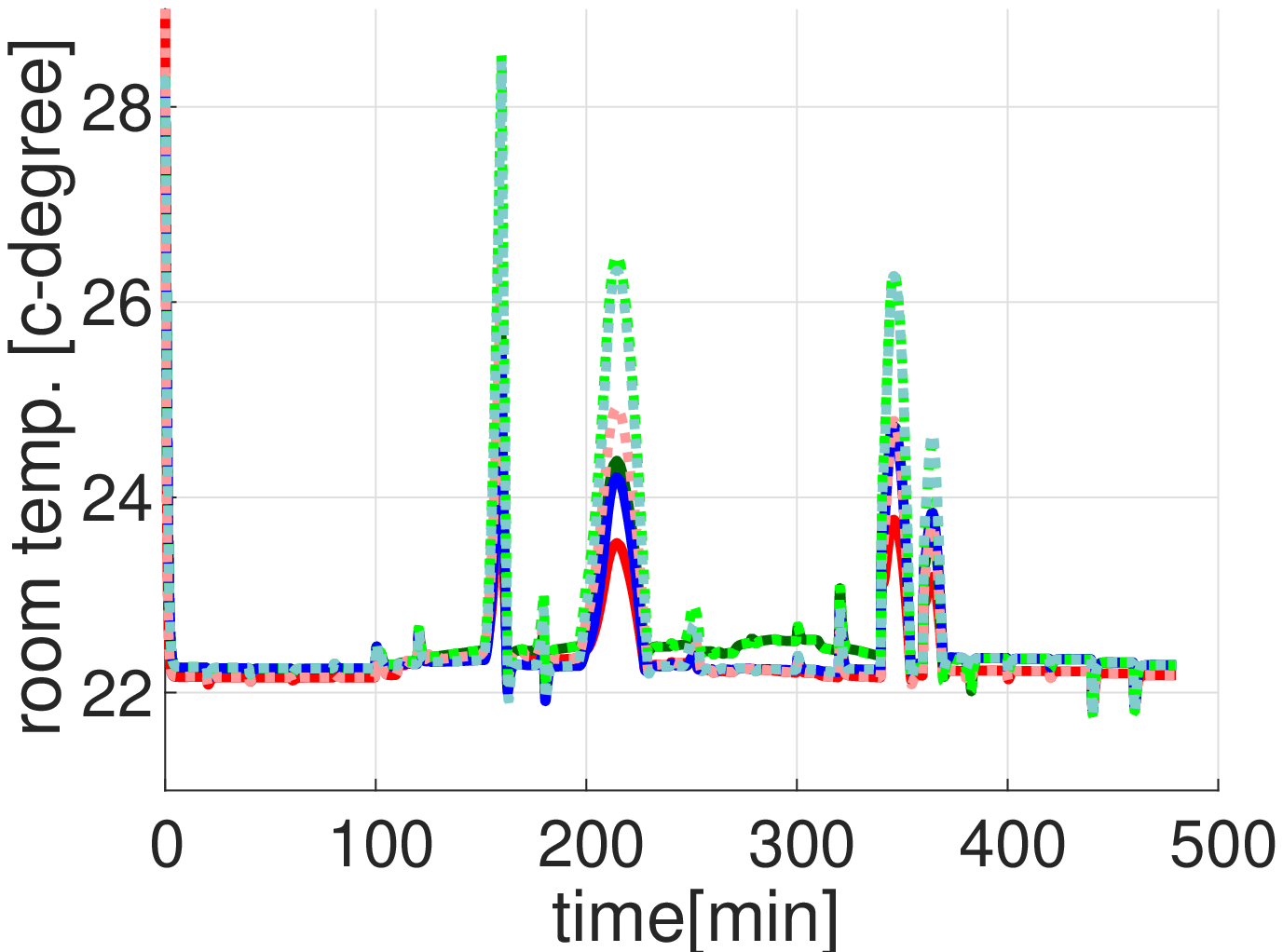}
\end{minipage}
\caption{Time responses of the estimated heat gain (top 2 figures), 
the estimated optimal room temperatures (bottom-left) and actual room temperatures (bottom-right).
In all figures, the solid curves show the responses delivered by the proposed method,
and the dotted ones with light colors are those by the disturbance feedforward scheme.
In the top figures, the dotted lines coincide with the actual heat gain.}
\label{fig:18}  
\end{figure}

In the following simulation, 
we use the disturbance data shown in Fig. \ref{fig:15}.
Here, we compare the results with the ideal case that the disturbance
$w_{\rm q}$ is directly measurable. 
In this case, the feedback path
from the physical dynamics to optimization is not needed and
hence we take the cascade connection from optimization to physics.
It is to be noted that it is hard to implement this in practice.

\begin{figure}[t]
\centering
\begin{minipage}[b]{4.2cm}
\includegraphics[width=4.2cm]{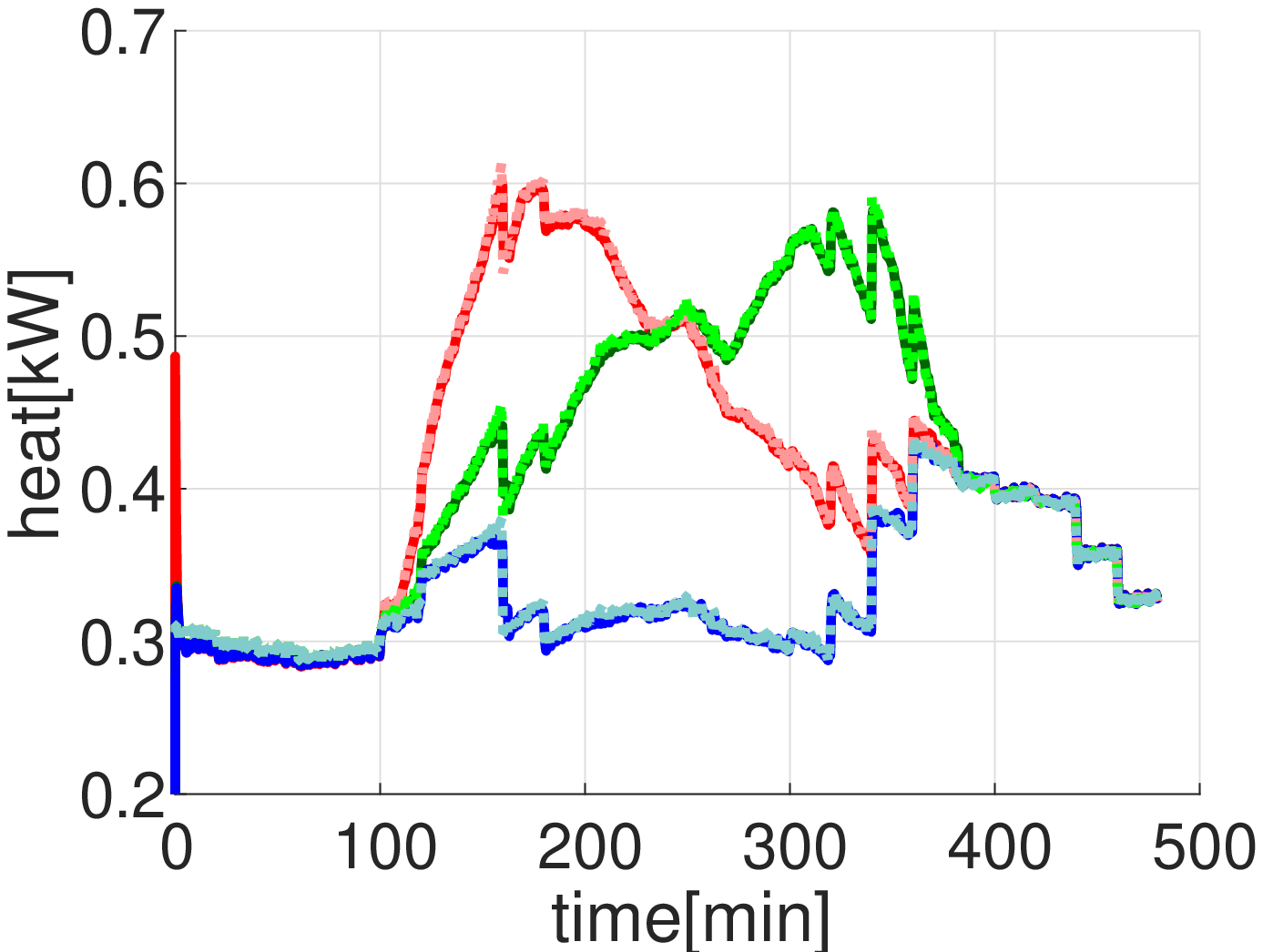}
\end{minipage}
\begin{minipage}[b]{4.2cm}
\includegraphics[width=4.2cm]{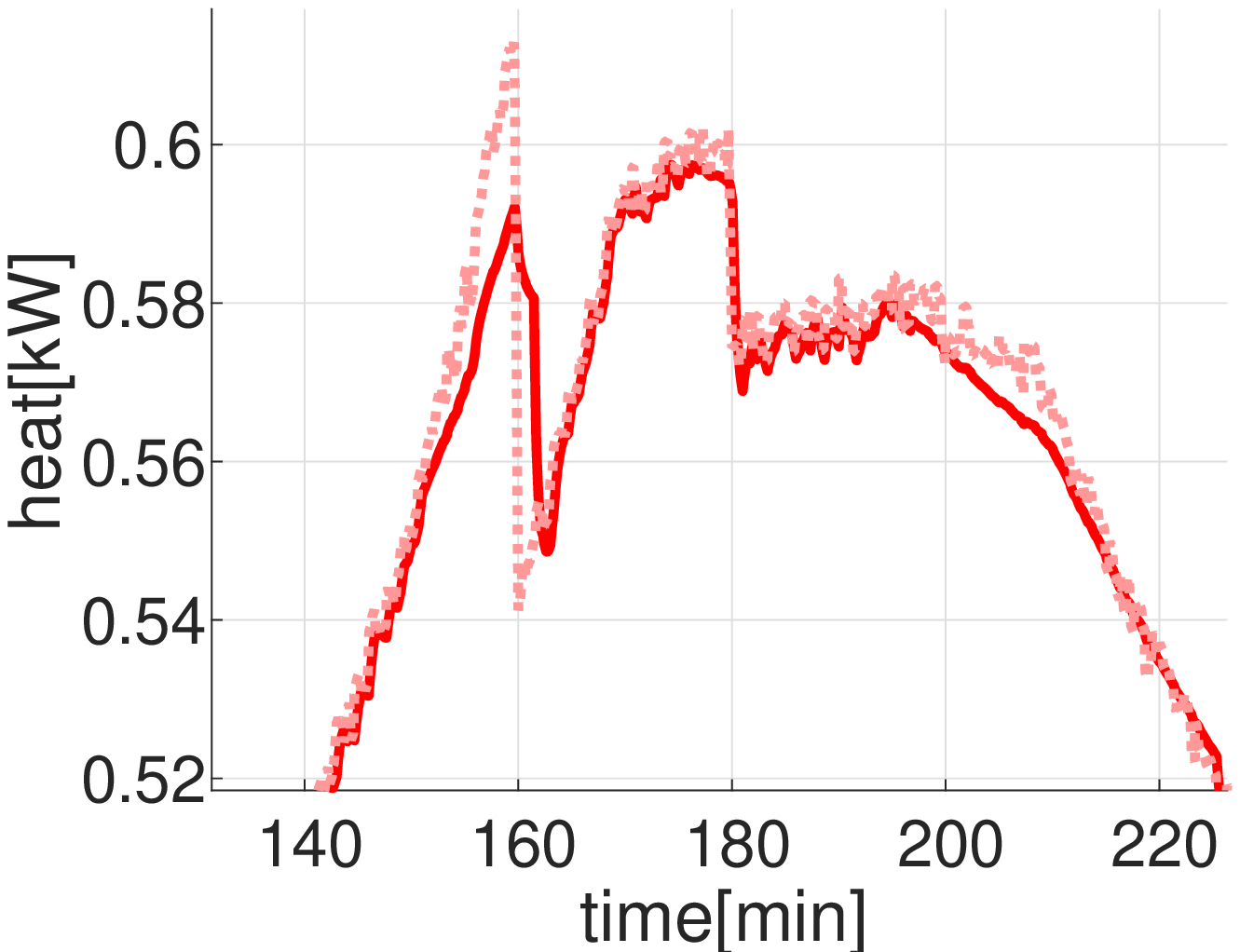}
\end{minipage}
\medskip

\begin{minipage}[b]{4.2cm}
\includegraphics[width=4.2cm]{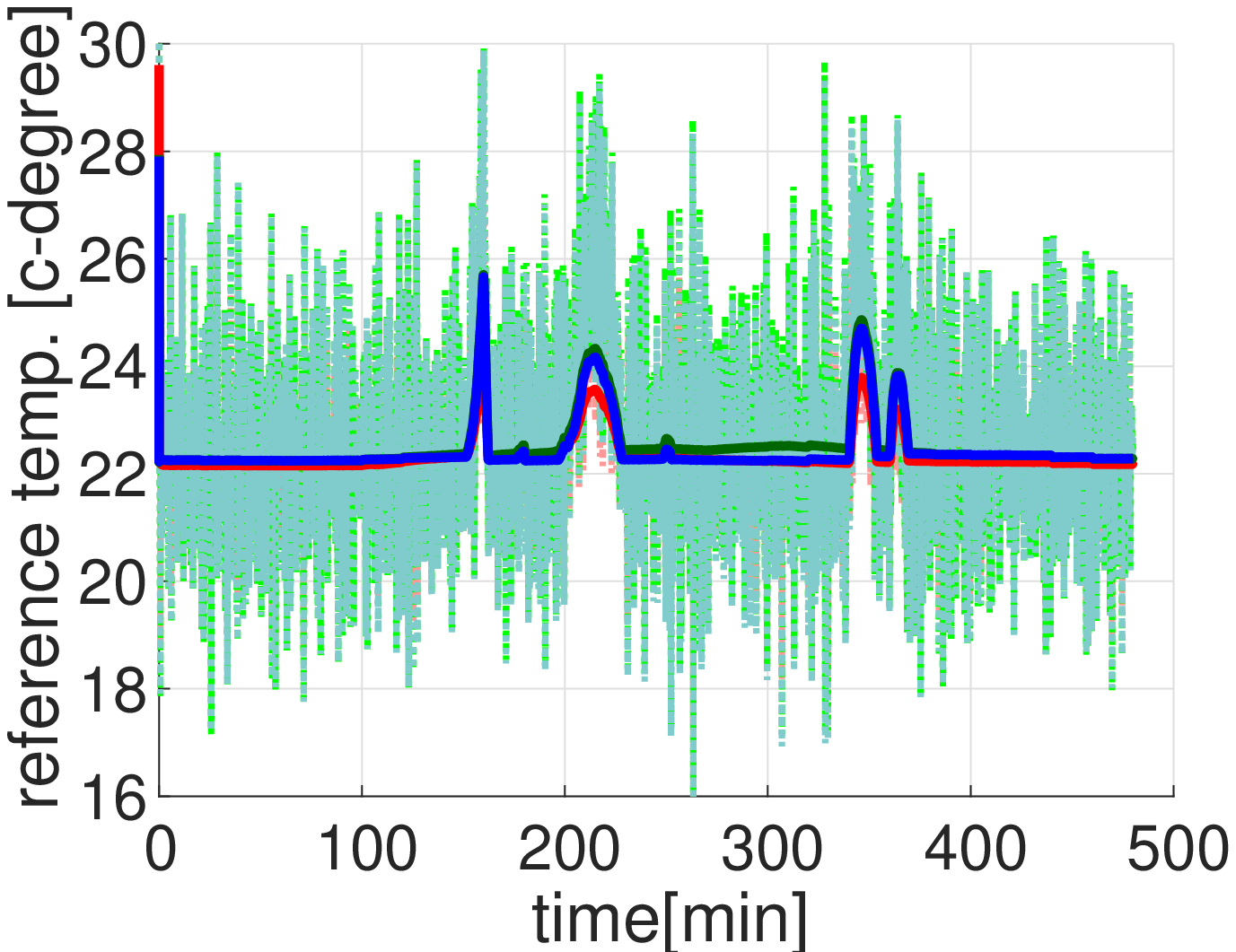}
\end{minipage}
\begin{minipage}[b]{4.2cm}
\includegraphics[width=4.2cm]{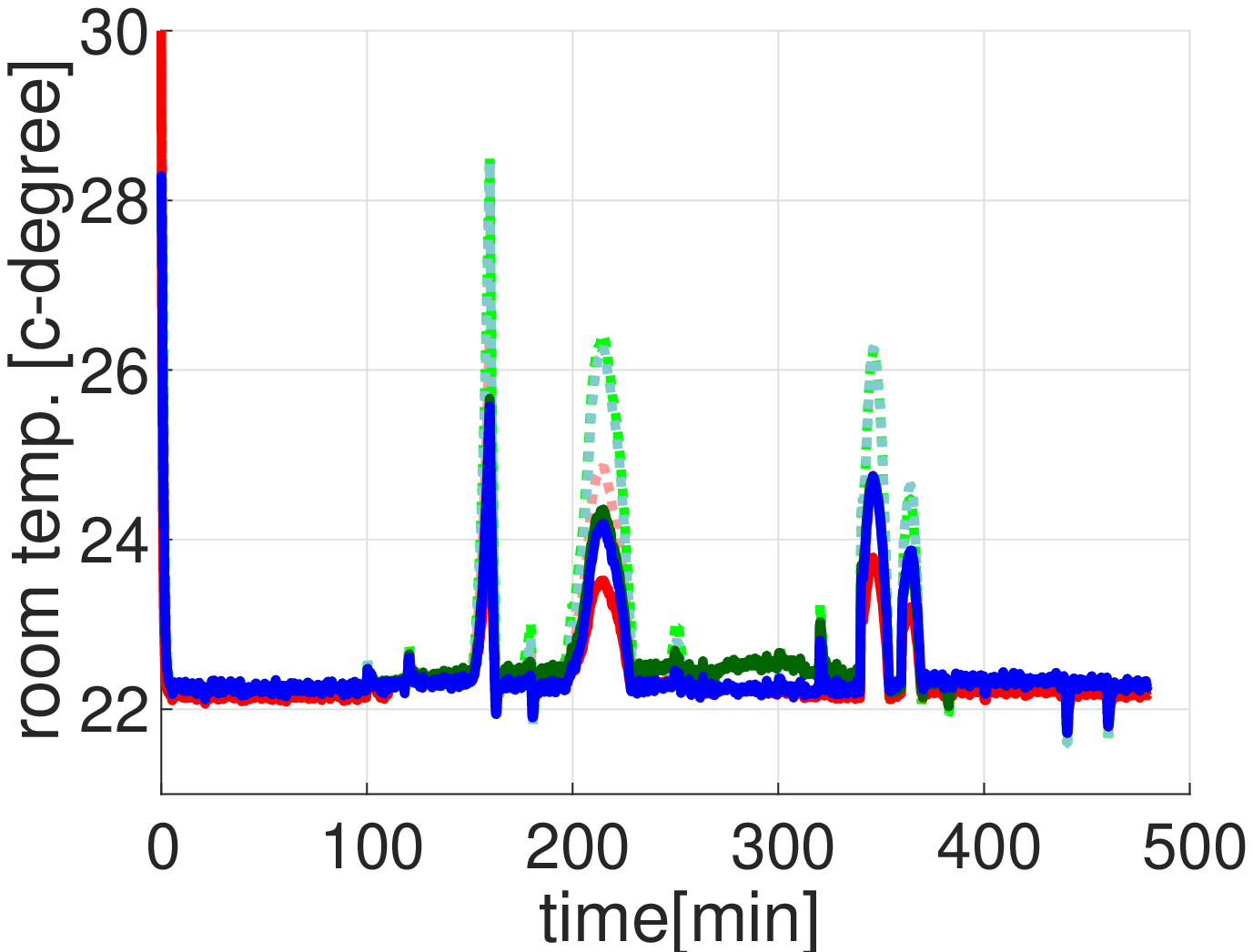}
\end{minipage}
\caption{Time responses in the presence of noises on the heat gain $w_{\rm q}$,
where every line has the same meaning as Fig. \ref{fig:18}.}
\label{fig:17}
\medskip

\begin{minipage}[b]{4.2cm}
\includegraphics[width=4.2cm]{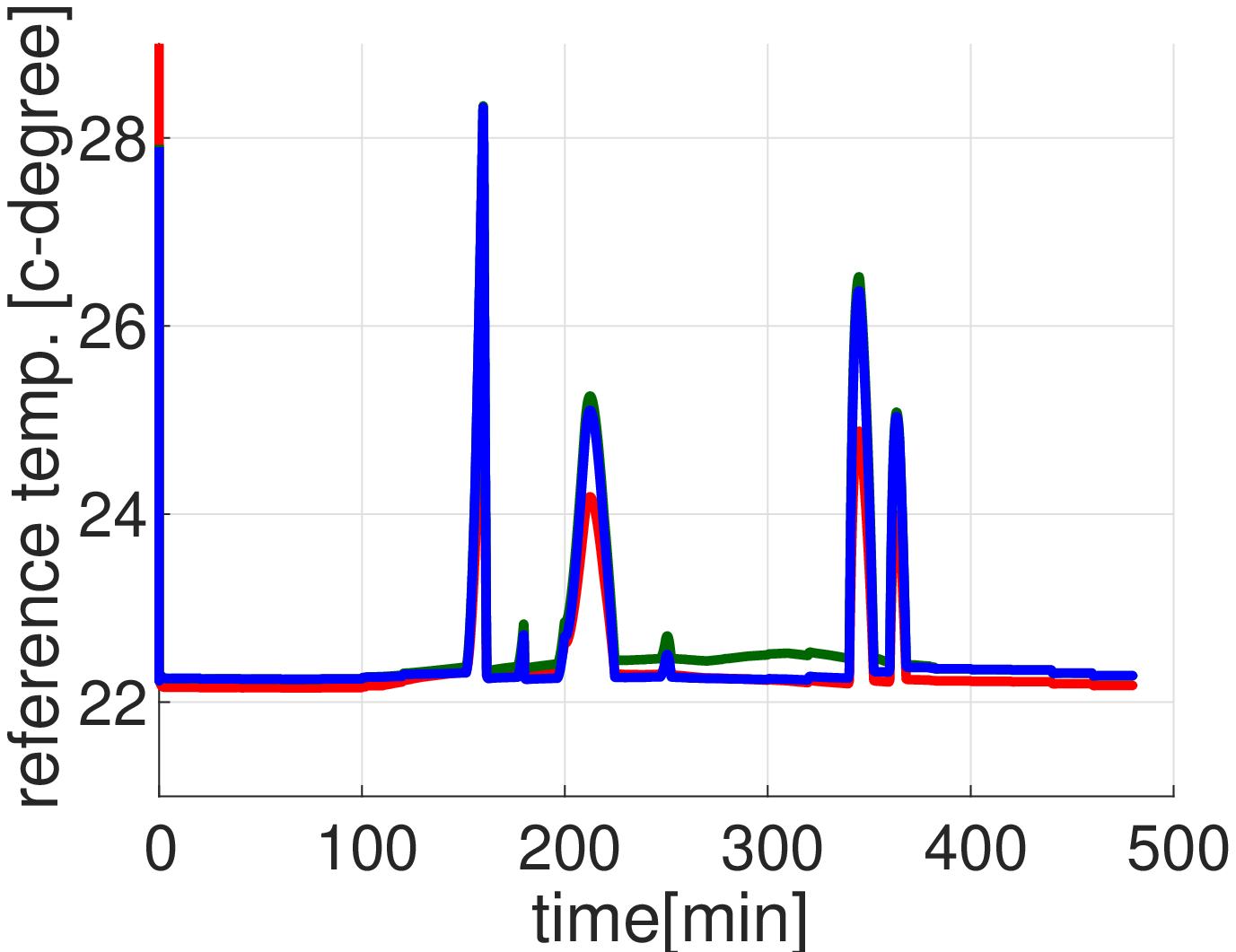}
\end{minipage}
\begin{minipage}[b]{4.2cm}
\includegraphics[width=4.2cm]{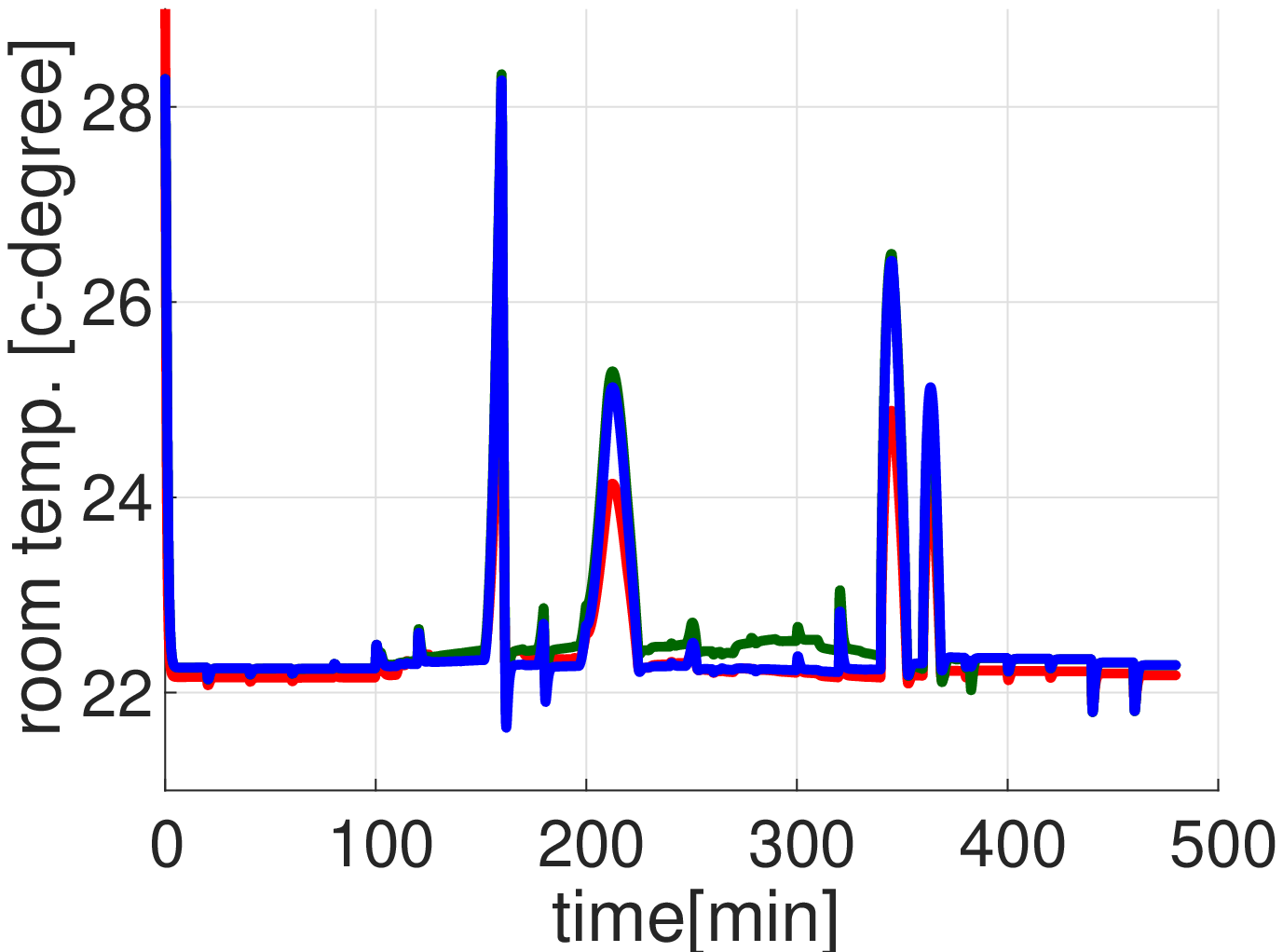}
\end{minipage}
\caption{Time responses of the reference (left) and room temperatures (right)
with the scaling factor $\theta = 50$.}
\vspace{-.2cm}
\label{fig:29}  
\end{figure}

The trajectories of the estimated heat gains,
the estimated optimal room temperatures
and temperatures $T_1, T_2, T_3$
are illustrated by solid curves in Fig. \ref{fig:18}.
The dotted lines with light colors show the trajectories using the above disturbance feedforward scheme.
We see from the top left figures that the presented algorithm almost
correctly estimate the disturbance $w_q$.
It is also observed from the right fine-scale figure that 
high frequency signal components are filtered out in the case of our algorithm.
The trajectories in the bottom-left figure sometimes get far from 22${}^{\circ}$C since
the constraint $\sum_{i=1}^3|z_{ui}| \leq 1.25$ gets active during the periods
due to the high ambient temperature and heat gains.
We see from the bottom figures that
the response of the present method to the constraint violations is slower
than the disturbance feedforward scheme because
high frequency components of the disturbance are filtered out by the physical dynamics.

Let us next show that the high sensitivity of disturbance feedforward 
can cause another problem.
Here, small noises are added to the heat gain $w_{\rm q}$ at every 20s,
whose absolute value is upper bounded by $1.0\times 10^{-3}$.
We then run the above two algorithms.
The resulting responses are illustrated in Fig. \ref{fig:17}.
It is observed from the top-left figure that the noise is filtered out
in the present algorithm, and its effects do not appear on the estimates.
Accordingly, the trajectories of the estimated optimal and actual room temperatures are almost the same as
Fig. \ref{fig:18}.
Meanwhile, the disturbance feedforward approach suffers significant
effects from the noise. 
The trajectories of the bottom-left get smaller than 22${}^{\circ}$C.
Namely, the fluctuations are caused 
by the undesirable over- and undershoots.
Although the trajectories of the actual temperatures get smooth, this behavior  of the reference
is not desirable from an engineering point of view.
Adding a low-pass filter or reducing the gain of the optimization dynamics
would eliminate the oscillations but it spoils the advantage, namely response speed.
It is to be noted that if the disturbance feedforward is implemented using
the recovery technique in \cite{Xuan},
the low-pass filter is not designed independently
of system stability 
as stated in Section \ref{sec:5}.

If the response speed of the present algorithm in
Fig. \ref{fig:18} is still problematic, it can be accelerated by tuning the scaling factor $\theta$.
The results for $\theta = 50$ are shown in Fig. \ref{fig:29}, where it is observed that
almost the same speed as the disturbance feedforward in Fig. \ref{fig:18} is achieved by the present algorithm. 
Simulation for a larger-scale system with more practical settings is left as a future work of this paper.

\section{Conclusion}

In this paper, we presented 
a novel combined optimization and control algorithm
for HVAC control of buildings.
We designed a primal-dual algorithm-based optimization dynamics and
a local physical control system, and proved the system properties related  to passivity. 
We then interconnected the optimization and
physical dynamics, and proved convergence of the room temperatures to
the optimal ones.
We finally demonstrated the present algorithms through simulation.



\appendices

\section{Proof of Lemma \ref{lem:3}}
\label{appendix:1}

\begin{lemma}
\label{lem:hata1}
Consider the system (\ref{eqn:3.4b}) with $\hat \lambda(0) \geq 0$. 
Then, under Assumption \ref{ass:0}, it is passive from $\tilde z_u = \hat z_u - z_u^*$ to $\tilde p = p - p^*$ with $p^* := \nabla g(z^*_u)\lambda^*$.
\end{lemma}
\begin{proof}
Define the energy function
$U := \frac{1}{2}\|\hat \lambda - \lambda^*\|^2$.
Then, following the same procedure as
\cite{hatanaka2}, we have
\begin{eqnarray}
D^+ U \leq
 (p - p^*)^{\top}(\hat z_u - z^*_u) = \tilde p^{\top} \tilde z_u,
\label{eqn:3.14}
\end{eqnarray}
where the notation $D^+$ represents the upper Dini derivative.
Integrating this in time completes the proof.
\end{proof}

We next consider  (\ref{eqn:3.4a}). 
Now, replace $-M^2(\hat z_u + \hat d_{\rm q}) - p$ by an external input $\mu$ and
consider the system 
\begin{eqnarray}
\dot{\hat{z}}_u =  - \alpha\{N(w_{\rm a} - \bar h) + \nabla f(\hat z_u) - \mu\}.
\label{eqn:3.19a}
\end{eqnarray}
Then, we have the following lemma.

\begin{lemma}
\label{lem:hata2} 
Suppose $w_{\rm a} \equiv d_{\rm a}$.
Then, under Assumption \ref{ass:0}, the system  (\ref{eqn:3.19a}) is passive from $\tilde \mu := \mu - \mu^*$ to $\tilde z_u = \hat z_u - z_u^*$,
where $\mu^* := -M^2(z^*_u + d_{\rm q})-p^*$.
\end{lemma}

\begin{proof}
Subtracting (\ref{eqn:3.5a}) from (\ref{eqn:3.19a}) under yields
\begin{eqnarray}
\dot{\tilde{z}}_u =  -\alpha (\nabla f(\hat z_u) - \nabla f(z_u^*)) + \alpha \tilde \mu.
\label{eqn:3.9a}
\end{eqnarray}
Now, define 
$V := \frac{1}{2\alpha}\|\tilde z_u\|^2 = \frac{1}{2\alpha}\|\hat z_u - z_u^*\|^2$.
Then, the time derivative of $V$ along the trajectories of (\ref{eqn:3.9a}) is given by
\begin{eqnarray}
\dot V = - (\hat z_u - z_u^*)^{\top}(\nabla f(\hat z_u) - \nabla f(z_u^*))  + \tilde z_u^{\top}\tilde \mu.
\label{eqn:3.16}
\end{eqnarray}
From convexity of $f$, $(\hat z_u - z_u^*)^{\top}(\nabla f(\hat z_u) - \nabla f(z_u^*)) \geq 0$ holds \cite{boyd}.
This completes the proof.
\end{proof}

The system (\ref{eqn:3.4}) is given by interconnecting  
(\ref{eqn:3.4b}) and (\ref{eqn:3.19a}) via $\mu = \nu - p$. 
It is then easy to confirm that $\tilde \mu = \tilde \nu - \tilde p$.

We are now ready to prove Lemma \ref{lem:3}.
Define $S_{\rm o} : = V+U$.
Then, combining (\ref{eqn:3.14}), (\ref{eqn:3.16}) and $\tilde \mu = \tilde \nu - \tilde p$, we have
\begin{eqnarray}
D^+ S_{\rm o} \leq - (\hat z_u - z_u^*)^{\top}(\nabla f(\hat z_u) - \nabla f(z_u^*)) + \tilde z_u^{\top}\tilde \nu.
\label{eqn:3.17}
\end{eqnarray}
Since $\tilde \nu = -M^2(\tilde z_u + \tilde d_{\rm q})$,
it follows
\begin{align}
D^+ S_{\rm o} &\leq
- \tilde \nu^{\top}\tilde d_{\rm q} - (\tilde z_u + \tilde d_{\rm q})^{\top}M^2(\tilde z_u + \tilde d_{\rm q})
\nonumber\\
&
- (\hat z_u - z_u^*)^{\top}(\nabla f(\hat z_u) - \nabla f(z_u^*))\leq -\tilde \nu^{\top}\tilde d_{\rm q}.
\label{eqn:3.18}
\end{align}
This completes the proof.

\section{Proof of Lemma \ref{lem:hata12}}
\label{appendix:2}

\begin{lemma}
\label{ass:2}
Under Assumption \ref{ass:1} and (\ref{eqn:9.1}), the system $(\bar A_3, A_2P^{-1/2})$ is stabilizable and
$(P^{-1/2}MA_2^{\top},\bar A_3)$ is detectable, where $\bar A_3 := -A_3 + A_2MP^{-1}A_2^{\top}$.
\end{lemma}

\begin{proof}
Define $\Phi_{\rm s}:= P^{1/2}MP^{-1}A_2^{\top}$ and $\Phi_{\rm d}:= A_2MP^{-1}M^{-1}P^{1/2}$.
Then, 
\begin{align}
\bar A_3 - A_2P^{-1/2}\Phi_{\rm s}=-A_3,\
\bar A_3  - \Phi_{\rm d}P^{-1/2}MA_2^{\top} = -A_3
\nonumber
\end{align}
hold and $-A_3$ is stable. This completes the proof.
\end{proof}

Using Lemma \ref{ass:2},
we next prove the following result.
\begin{lemma}
\label{lem:hata11}
Consider the system (\ref{eqn:99.5a})
with $r \equiv r^*$,
$w_{\rm a} \equiv d_{\rm a}$ and $w_{\rm q} \equiv d_{\rm q}$.
Then, under Assumption \ref{ass:1}, the system is 
passive from $\tilde \zeta$ to 
$\tilde x_1 := B^{\top} \tilde x$ with $\tilde x := x - x^*$.
\end{lemma}

\begin{proof}
We first formulate the error system
\begin{subequations}
\label{eqn:99.7}
\begin{eqnarray}
\hspace{-.7cm}
\dot{\tilde{x}} \!\!&\!\!=\!\!&\!\!  
-\bar A \tilde x + BM^{-1}\tilde \zeta
\label{eqn:99.7a}\\
\hspace{-.7cm}
\dot{\tilde{\xi}} \!\!&\!\!=\!\!&\!\! k_{\rm I}(\tilde r - B\tilde x),\
\tilde\zeta =
\bar k_{\rm P} M (\tilde r  - B^{\top} \tilde x) 
+ M\tilde \xi
\label{eqn:99.7b}
\end{eqnarray}
\end{subequations}
where $\tilde \xi := \xi - \xi^*$.
Take a positive definite matrix $\Psi \in \R^{n_2\times n_2}$ and define
$\bar \Psi := \begin{bmatrix}
M&0\\
0&\Psi
\end{bmatrix}\in \R^{n\times n}$.
Then, by calculation, we have
\begin{eqnarray}
\bar \Psi \bar A + \bar A \bar \Psi = 
\begin{bmatrix}
P&MA_2^{\top}+A_2^{\top}\Psi \\
\Psi A_2+A_2 M&\Psi A_3 + A_3\Psi
\end{bmatrix}.
\label{eqn:9.8}
\end{eqnarray}
From Schur complement, under 
Assumption \ref{ass:1}, $\bar \Psi \bar A + \bar A \bar \Psi > 0$ is equivalent to the following Riccati inequality.
\begin{eqnarray}
 - \bar A_3\Psi -\Psi \bar A_3^{\top}
+ \Psi A_2 P^{-1}A_2^{\top}\Psi + A_2 M P^{-1} MA_2^{\top} < 0
\label{eqn:99.9}
\end{eqnarray}
A positive semi-definite solution $\Psi$ to (\ref{eqn:99.9}) is shown to exist
from Lemma \ref{ass:2}.
Now, define an energy function $S_x := \frac{1}{2}\tilde x^{\top}\bar \Psi \tilde x$
for the solution $\Psi$ to (\ref{eqn:99.9}).
Then, the time derivative of $S_x$ along the trajectories of (\ref{eqn:99.7a})
is given by
\begin{eqnarray}
\dot S_x \!\!&\!\!=\!\!&\!\! 
- \frac{1}{2}\tilde x (\bar \Psi \bar A + \bar A \bar \Psi) \tilde x 
+ \tilde x^{\top} \bar \Psi BM^{-1}\tilde \zeta
\nonumber\\
\!\!&\!\!\leq\!\!&\!\! 
\tilde x^{\top} B \tilde \zeta = (B^{\top}\tilde x)\tilde \zeta = \tilde x_1^{\top}\tilde \zeta.
\label{eqn:9.10}
\end{eqnarray}
This completes the proof.
\end{proof}

We are now ready to prove Lemma \ref{lem:hata12}.
Replace $\tilde r - B^{\top}\tilde x$ in  (\ref{eqn:99.7b}) by $\tilde e$ as
\begin{eqnarray}
\dot{\tilde{\xi}} = k_{\rm I}\tilde e,\
\tilde\zeta =
\bar k_{\rm P} M \tilde e
+ M\tilde \xi.
\label{eqn:9.12}
\end{eqnarray}
Define $S_{\xi} := \frac{1}{2k_{\rm I}}\tilde \xi^{\top} M\tilde \xi$.
Then, the time derivative of $S_{\xi}$ along the trajectories of 
(\ref{eqn:9.12}) is given as
\begin{eqnarray}
\dot S_{\xi} =  \tilde \xi^{\top}M\tilde e
= (\tilde \zeta - \bar k_{\rm P} M \tilde e)^{\top}\tilde e
=\tilde \zeta^{\top}\tilde e - k_{\rm P} \tilde e^{\top}M\tilde e. 
\label{eqn:99.13}
\end{eqnarray}
Define $S_{\rm p} := S_x + S_{\xi}$.
Then, from (\ref{eqn:9.10}) and (\ref{eqn:99.13}),
we have
\begin{eqnarray}
\dot S_{\rm p} \!\!&\!\!\leq\!\!&\!\!  
\tilde \zeta^{\top}\tilde r -  k_{\rm P} (\tilde r - \tilde x_1)^{\top}M(\tilde r - \tilde x_1)
\leq \tilde \zeta^{\top}\tilde r. 
\label{eqn:9.14}
\end{eqnarray}
This completes the proof.

\end{document}